\documentclass{tlp}

\usepackage{amsmath}
\usepackage{graphicx}
\usepackage{multirow}

\usepackage{amsfonts}
\usepackage{amssymb}
\usepackage{amsmath}
\usepackage{latexsym}
\usepackage{stmaryrd}
\usepackage{microtype}
\usepackage{comment}
\usepackage{tikz-cd}
\usepackage{preamble}

\newtheorem{theorem}{Theorem}
\newtheorem{definition}{Definition}
\newtheorem{proposition}{Proposition}
\newtheorem{example}{Example}
\newtheorem{corollary}{Corollary}
\newtheorem{lemma}{Lemma}

\makeatletter
\newcommand\qed{\hspace*{1em}\hbox{\proofbox}\endtrivlist} 
\newenvironment{proofwithoutqed}
{\@ifnextchar[{\@oprf}{\@nprf}}
{}
\makeatother

\newcommand{\op}[1]{{#1}^{\mathit{op}}}

\newcommand{\catname}[1]{{\normalfont\textbf{#1}}}
\newcommand{\CLat}{\catname{CLat}}

\newcommand{\CPO}{\catname{CPO}}
\newcommand{\poset}{\catname{POSet}}
\newcommand{\BiLat}{\catname{BiLat}}

\newcommand{\Approx}{\catname{Approx}}	
\newcommand{\greeks}{\catname{LUcons}}	
\newcommand{\compjoinsl}{\catname{CJSLat}}	
	
\newcommand{\Obj}[1]{\normalfont\text{Ob}(\text{#1})}
\newcommand{\Hom}[1]{\normalfont\text{Mor}(\text{#1})}

\newcommand{\leqpt}{\leq_\mathit{pt}}

\newcommand{\leqpr}{\leq} 
\newcommand{\zeroproj}[1]{\mathfrak{p}^0_{#1}}
\newcommand{\famtypes}{\tau} 
\newcommand{\subscript}[2]{{#1}_{\!{#2}}}
\newcommand{\HOL}{\mathcal{HOL}} 

\DeclareMathOperator\lfp{lfp}

\newcommand*{\pnot}{\mathord{\sim}}

\newcommand{\lub}{\bigsqcup}

\newcommand{\glb}{\bigsqcap}

\newcommand\leqp{\leq_p}
\newcommand\lrule{\leftarrow}
\newcommand\ltrue{\mathbf{t}}
\newcommand\lfalse{\mathbf{f}}

\newcommand{\oldmwrst}[3]{\llbracket#1\rrbracket_{#3}(#2)}

\newcommand{\exact}[1]{\mathcal{E}_{#1}}

\usepackage{xspace}
\newcommand\m[1]{\ensuremath{#1}\xspace}
\newcommand\basedom{\m{\iota}}
\newcommand\bool{\m{o}}

\usepackage[disable]{review}

\usepackage[final]{listings}
\newcommand\newlineInListing{\mbox{\textcolor{red}{$\hookrightarrow$}\space}}

\newcommand\ie{i.e.,\xspace}

\lstdefinelanguage{holp}{
	morekeywords=[1]{X,Y,Z,M,T}, 
	morekeywords=[2]{P,Q,R,S}, 
	morekeywords=[3]{A,B,C}, 
	morecomment=[l]{\%}
}
\begin{document}
	
	\lefttitle{Cambridge Author}
	
	\jnlPage{1}{8}
	\jnlDoiYr{2021}
	\doival{10.1017/xxxxx}
	
	\title[]{A Category-Theoretic Perspective on Higher-Order Approximation Fixpoint Theory\thanks{This study was funded
			by Fonds Wetenschappelijk Onderzoek -- Vlaanderen (project G0B2221N, and grant V426524N), and by the European Union -- NextGenerationEU under the National Recovery and Resilience Plan ``Greece 2.0'' (H.F.R.I. ``Basic research Financing (Horizontal support of all Sciences)'', Project Number: 16116).}}
	
	\begin{authgrp}
		\author{\sn{Pollaci} \gn{Samuele}}
		\affiliation{Vrije Universiteit Brussel, Belgium \and
			Katholieke Universiteit Leuven, Belgium }
		\author{\sn{Kostopoulos} \gn{Babis}}
		\affiliation{Harokopio University of Athens, Greece}
			\author{\sn{Denecker} \gn{Marc}}
		\affiliation{Katholieke Universiteit Leuven, Belgium }
			\author{\sn{Bogaerts} \gn{Bart}}
		\affiliation{
			Katholieke Universiteit Leuven, Belgium \and Vrije Universiteit Brussel, Belgium 
			}
	\end{authgrp}
	
	\history{\sub{xx xx xxxx;} \rev{xx xx xxxx;} \acc{xx xx xxxx}}
	
	\maketitle
	
	\begin{abstract}
	Approximation Fixpoint Theory (AFT) is an algebraic framework designed to study the semantics of non-monotonic logics. Despite its success, AFT is not readily applicable to higher-order definitions. To solve such an issue, we devise a formal mathematical framework employing concepts drawn from Category Theory. In particular, we make use of the notion of  Cartesian closed category to inductively construct higher-order approximation spaces while preserving the structures necessary for the correct application of AFT. We show that this novel theoretical approach extends standard AFT to a higher-order environment, and generalizes the AFT setting of  \cite{CRS18ApproximationFixpointTheoryWell-FoundedSemanticsHigher-Order}. 
	
	Under consideration in Theory and Practice of Logic Programming (TPLP).
	\end{abstract}
	
	\begin{keywords}
		Approximation fixpoint theory, Higher-order definitions, Category theory.
	\end{keywords}

	\section{Introduction}\label{sec:Introduction}
	
	Approximation Fixpoint Theory (AFT) \citep{DMT00ApproximationsStableOperatorsWell-FoundedFixpointsApplications}
is an algebraic framework designed to study the semantics of non-monotonic logics. 
It was originally designed for characterizing the semantics of logic programming, autoepistemic logic, and default logic, and to resolve longstanding problems on the relation between these formalisms \citep{DMT11ReitersDefaultLogicLogicAutoepistemicReasoning}. Later, it has also been applied to a variety of other domains, including abstract argumentation \citep{S13Approximatingoperatorssemanticsabstractdialecticalframeworks,B19WeightedAbstractDialecticalFrameworksthroughLens},  active integrity constraints \citep{BC18Fixpointsemanticsactiveintegrityconstraints}, stream reasoning \citep{A20Fixedpointsemanticsstreamreasoning},  integrity constraints for the semantic web \citep{BJ21FixpointSemanticsRecursiveSHACL}, and Datalog \citep{pollaci2025}.

The core ideas of AFT are relatively simple: we are interested in fixpoints of an operator $O$ on a given lattice $\langle L,\leq\rangle$. 
For monotonic operators, Tarski's theory guarantees the existence of a least fixpoint, which is of interest in many applications. 
For non-monotonic operators, the existence of fixpoints is not guaranteed; and even if fixpoints exist, it is not clear which would be ``good'' fixpoints.
AFT generalizes Tarki's theory for monotonic operators by making use of a so-called \emph{approximating operator}; this is an operator $A: L^2\to L^2$, that operates on $L^2$, and that is monotonic with respect to the precision order $\leqp$ (defined by $(x,y)\leqp(u,v)$ if $x\leq u$ and $v \leq y)$). The intuition is that elements of $L^2$ approximate elements of $L$:  $(x,y)\in L^2$ approximates $z$ if $x\leq z\leq y$, \ie when $x\leq y$, the tuple $(x,y)$ can be thought of as an interval in $L$. 
Given such an approximator, AFT defines several types of fixpoints (supported fixpoints, a Kripke-Kleene fixpoint, stable fixpoints, a well-founded fixpoint) of interest. 

In several fields of non-monotonic reasoning, 
it is relatively straightforward to define an approximating operator 
and
it turns out that the different types of fixpoints then correspond to existing semantics. 
In this way, AFT clarifies on the one hand how different semantics in a single domain relate, and on the other hand what the relation is between different (non-monotonic) logics. 

Let us illustrate the application of AFT to standard, first-order, logic programming. 
In this setting, the lattice $L$ is the lattice of interpretations with the truth order $I\leq J$ if $P^I\subseteq P^J$ for each predicate $P$. 
The operator is the immediate consequence operator $T_P$, as defined in the seminal work of 
\citet{vK76SemanticsPredicateLogicProgrammingLanguage}. 
Given a logic program (\ie a set of rules), this operator has the property that $q$ holds in $T_P(I)$ if and only if there is a rule 
$q \lrule \varphi $
in $P$ such that $\varphi$ is true in $I$. 
In this setting, pairs $(I,J)$ are seen as \emph{four-valued} interpretations: $I$ represents what is \emph{true} and $J$ what is \emph{possible}. 
A fact is then \emph{true} (resp.~\emph{false}) if it is true (resp.~false) in both $I$ and $J$, \emph{unknown} if it is true in $J$ but not true in $I$ and \emph{inconsistent} if it is true in $I$ but not in $J$. 
The approximating operator $\Psi_P$ is, in this case, nothing else than Fitting's (\citeyear{F02Fixpointsemanticslogicprogrammingsurvey}) four-valued immediate consequence operator, 
which uses Kleene's truth tables to evaluate the body of each rule in a four-valued interpretation. 
For this approximator, the fixpoints defined by AFT correspond to the major semantics of logic programming: supported fixpoints are models of Clark's completion \citep{C77NegationFailure}, stable fixpoints correspond to (partial) stable models \citep{GL88StableModelSemanticsLogicProgramming}, the Kripke-Kleene fixpoint to the Kripke-Kleene model \citep{F85Kripke-KleeneSemanticsLogicPrograms} and the well-founded fixpoint is the well-founded model 
\citep{VRS91Well-FoundedSemanticsGeneralLogicPrograms}. 

This paper is motivated by a need to apply AFT to \emph{higher-order} logic programming that arose in several contexts \citep{DvJD15Semanticstemplatescompositionalframeworkbuildinglogics,DvBJD16CompositionalTypedHigher-OrderLogicDefinitions,CRS18ApproximationFixpointTheoryWell-FoundedSemanticsHigher-Order}.   
An important issue that arises in this context is that using pairs of interpretations no longer allows for an obvious way to evaluate formulas in an approximation. 
Let us illustrate this with a brief example (for more detailed ones, we invite the reader to  look at Examples \ref{ex:graph}, and \ref{ex:manifacturer}).
Consider a logic program in which a first-order predicate $p$ and a second-order predicate $Q$ are defined. 
Now assume that in the body of a rule, the atom $Q(p)$ occurs.   
A tuple $(I,J)$ of interpretations in this case tells us whether $Q(S)$ is true, false, unknown, or inconsistent, for any given set $S$.
However, the interpretation of $p$ via $(I,J)$ is not a set, but a partially defined set, making it hard to evaluate expressions of the form $Q(p)$.
In other words, 
an approximation of the interpretation of $Q$ has to take as argument not only sets, \ie  \emph{exact} elements, but also partially defined sets, \ie \emph{approximate} elements, like the interpretation of $p$ in this example. 
Thus, 
there is a need for a richer space of approximations where approximate objects can be applied to other approximate objects. 

The above example and considerations suggest that spaces of approximations of higher-order objects should be defined inductively from lower-order ones, following the type hierarchy: we start by assigning a \emph{base approximation space} to each type at the bottom of the hierarchy, and then, for each composite type $\tau_1 \to \tau_2$, we define its approximation space as \emph{a certain class of functions} from the approximation space for $\tau_1$ to the approximation space for $\tau_2$, and so on.
This method was heavily inspired by the approach used by
\citet{CRS18ApproximationFixpointTheoryWell-FoundedSemanticsHigher-Order} to obtain a generalization of the well-founded semantics for higher-order logic programs with negation.
Notice that there are two major points in the construction above which are yet not defined: the base approximation spaces, and the class of functions we consider. The main question of this paper is how to define them in a generic way that works in all applications of AFT.

We want to apply the same AFT techniques on approximation spaces at any hierarchy level, \ie on base approximation spaces and the aforementioned sets of functions, which should thus have the same algebraic structure. 
In  Category Theory (CT), the notion of \emph{Cartesian closed category} captures this behavior. 
A category consists of a collection of \emph{objects} and a collection of \emph{morphisms}, \ie relations between objects. For example, we can define the \emph{category of square bilattices} as the one having square bilattices as objects, and monotone functions as morphisms. The objects of a Cartesian closed category $\mathbf{C}$ satisfy a property that can be intuitively understood as follows: \emph{if $A$ and $B$ are two objects of $\mathbf{C}$, then the set of morphisms from $A$ to $B$ is also an object of $\mathbf{C}$}. Hence, if the base approximation spaces are objects of a Cartesian closed category, then the category contains the full hierarchy of spaces we are aiming for. 
We call such a Cartesian closed category an \emph{approximation category} and denote it by $\Approx$.

In this category-theoretic framework, the questions on the nature of the base approximation spaces and the class of functions reduce to defining 
	the objects and the morphisms of $\Approx$. 
	Clearly, this depends on the application we want to use AFT for. Different applications imply different higher-order languages, with different types, and possibly different versions of AFT (standard AFT \citep{DMT00ApproximationsStableOperatorsWell-FoundedFixpointsApplications}, consistent AFT \citep{DMT03Uniformsemantictreatmentdefaultautoepistemiclogics}, or other extensions \citep{CRS18ApproximationFixpointTheoryWell-FoundedSemanticsHigher-Order}). 
	To formalize this, and unify different AFT accounts, we develop the notion of an  \emph{approximation system}. Once a language and the semantics of its types are fixed, we can choose an approximation system that consists, among other things, of a Cartesian closed category $\Approx$, equipped with a function $\mathit{App}$ associating the semantics of a type to an approximation space in $\Approx$. 
	The approximation system also determines which elements of the approximation spaces are \emph{exact}, \ie which elements approximate exactly one element of the semantics of a type, and, for every type, it provides a projection from the exact elements to the objects they represent in the corresponding semantics. This is non-trivial for higher-order approximation spaces, and it is indeed fundamental to obtain a sensible account for AFT for higher-order definitions. 
	
	 In recent work, a stable semantics for higher-order logic programs was defined building on consistent AFT \citep{BCCKPR24StableModelSemanticsHigher-OrderLogicProgramming}. 
	In that work, the approach taken to evaluate an expression of the form $Q(p)$, instead of applying an approximate interpretation for $Q$ to an approximate interpretation for $p$, is to apply the approximate interpretation for $Q$ to \emph{all} exact interpretations for $p$ that are still possible, and returning the least precise approximation of all the results. 
	What this means in effect is that some sort of \emph{ultimate construction} \citep{DMT04Ultimateapproximationapplicationnonmonotonicknowledgerepresentation} is used; this has also been done in other extensions of logic programming \citep{PDB07Well-foundedstablesemanticslogicprogramsaggregates,DvBJD16CompositionalTypedHigher-OrderLogicDefinitions}. 
	\citet{BCCKPR24StableModelSemanticsHigher-OrderLogicProgramming} also pointed out a rather counterintuitive behaviour of the well-founded semantics defined in the work of \citet{CRS18ApproximationFixpointTheoryWell-FoundedSemanticsHigher-Order}, namely that even for simple non-recursive programs, the well-founded model might not assume expected values (leaving all atoms unknown). 
	It is important to mention, though, that this counterintuitive behaviour is caused solely by the treatment of (existentially) quantified variables and not by the algebraic theory (which is the focus of the current paper). Finally, it is interesting to mention that another paper \cite{CR23CategoricalApproximationFixpointTheory} joined CT and AFT, albeit with different perspective and aim: while we treat the whole set of possible approximation spaces as a category to apply any account of AFT to higher-order definitions, \citet{CR23CategoricalApproximationFixpointTheory} view the approximation spaces themselves as categories to provide a novel version of standard AFT.
	
In short, the main contributions of our paper are as follows: 
\begin{enumerate}
	\item We generalize the work of \citet{CRS18ApproximationFixpointTheoryWell-FoundedSemanticsHigher-Order} to a category-theoretic setting. In doing so, we shed light on the general principles underlying their constructions for higher-order logic programing and make their construction applicable to arbitrary current and future non-monotonic reasoning formalism captured by AFT. 
	\item We improve the work of \citet{CRS18ApproximationFixpointTheoryWell-FoundedSemanticsHigher-Order}. In particular, we define a new approximator, which provides the expected well-founded semantics; and we study the concept of \emph{exactness}, previously missing, allowing the use of the theory to define exact \emph{stable models} instead of focusing purely on the well-founded model. 
\end{enumerate}
	
	It is also worth remarking that the generality of the CT environment allows to cover various accounts of AFT, like the ones of \cite{DMT00ApproximationsStableOperatorsWell-FoundedFixpointsApplications},  \cite{DMT03Uniformsemantictreatmentdefaultautoepistemiclogics}, and \cite{CRS18ApproximationFixpointTheoryWell-FoundedSemanticsHigher-Order}, and possibly others. Different versions of AFT are suitable for various situations and cater to specific applications. For instance, consistent AFT \citep{DMT03Uniformsemantictreatmentdefaultautoepistemiclogics} utilizes a three-valued logic and provides a rather intuitive and easily applicable notion of approximation. On the other hand, standard AFT \citep{DMT00ApproximationsStableOperatorsWell-FoundedFixpointsApplications} employs a four-valued approach, with \emph{inconsistent} elements, as presented earlier in this introduction. This can sometimes be of a more difficult use in applications, but shows several advantages from the formal, mathematical standpoint: having a full bilattice, composed of both consistent and inconsistent elements, provides symmetry and allows for duality results to be derived, simplifying the proofs of the fundamental theorems at the core of this version of AFT. It is hence valuable to obtain a framework that covers as many accounts of AFT as possible.
	
	The rest of this paper is structured as follows. 
	In Section \ref{sec:preliminaries}, we provide an overview of the  fundamental concepts from AFT and CT that we use. Section \ref{sec:categoryApprox} presents the novel definitions of approximation system, with the category $\Approx$, and of exact elements of an approximation space. 
	In Section \ref{sec:standardAFT}, we show that the square bilattices form a Cartesian closed category that can be chosen as $\Approx$ for standard AFT. With a suitable choice of $\mathit{App}$ and exact elements, depending on the application at hand, we obtain an approximation system that recovers the framework of standard AFT and extends it to higher-order objects. This section can be skipped by the reader interested uniquely in the AFT version of \cite{CRS18ApproximationFixpointTheoryWell-FoundedSemanticsHigher-Order}, which is addressed in the following section.
	In Section \ref{sec:iclp18}, we apply the novel categorical framework to \cite{CRS18ApproximationFixpointTheoryWell-FoundedSemanticsHigher-Order}. First, in Subsection \ref{sec:sub:approxcat_greeks}, we show that the approximation spaces from \cite{CRS18ApproximationFixpointTheoryWell-FoundedSemanticsHigher-Order} form a Cartesian closed category. 
	Second, in Subsection \ref{sec:sub:approxsystem_greeks} we define an approximation system that 
	 enables us to reconstruct in a simple way (using the general principles outlined above) the semantic elements defined ad-hoc by \citet{CRS18ApproximationFixpointTheoryWell-FoundedSemanticsHigher-Order}. At the same time, the definition of such approximation system also resolves a question that was left open in that work. Namely, what we get now is a clear definition for exact higher-order elements, and, in particular, this allows to determine when a model of a program is two-valued  (see Example \ref{ex:exact_model_iclp18}). We proceed with Subsection \ref{sec:sub:approximator_greeks}, where we present the new approximator that adjusts the behaviour of the well-founded semantics of \cite{CRS18ApproximationFixpointTheoryWell-FoundedSemanticsHigher-Order} for programs with existential quantifiers in the body of rules. We close the subsection with two examples of  logic programs in which we need to apply an approximate object on another approximate object.
	 We conclude in Section~\ref{sec:conclusions}.

	\section{Preliminaries}\label{sec:preliminaries}
	
	In this section, we provide a concise 
	introduction to the formal concepts we utilize throughout the paper. We divide the content into two subsections. In the former (Section \ref{sec:AFT}), we outline the core ideas at the foundation of AFT, and we present in more detail the parts of the work of \citet{CRS18ApproximationFixpointTheoryWell-FoundedSemanticsHigher-Order} that we aim to modify in Section \ref{sec:iclp18}.
In the second subsection (Section \ref{sec:cattheory}), we present the notions of Category Theory (CT) we need, with the definition of Cartesian closed category being the key concept. For further information on CT, we refer to the book by
	\cite{R17Categorytheorycontext}.


	\subsection{Approximation Fixpoint Theory}\label{sec:AFT}
	
	
	AFT generalizes Tarki's theory to non-monotonic operators, with the initial goal of studying the semantics of non-monotonic logics. As such, AFT heavily relies on the following notions from order theory.
	
	A \emph{partially ordered set} (poset) $\mathcal{P}$ is a set equipped with a partial order, \ie a reflexive, antisymmetric, transitive relation. We denote a poset by $\mathcal{P}=\langle P, \leq_P \rangle$, where $P$ is the underlying set, and $\leq_P$ the partial order. By abuse of notation, 
	when referring to a poset $\mathcal{P}$, we often use the notation for the underlying set $P$ in place of the calligraphic one.
	We denote by $\op{\mathcal{P}}$ the poset with the same underlying set as $\mathcal{P}$ but opposite order, \ie $\op{\mathcal{P}}=\langle P, \geq_P\rangle$. 
	Given a subset $S\subseteq P$, a lower bound $l$ of $S$ is the \emph{greatest lower bound of $S$}, denoted by $\glb S$, if it is greater than any other lower bound of $S$. Analogously, an upper bound $u$ of $S$ is the \emph{least upper bound of $S$}, denoted by $\lub S$, 
	if it is lower than any other upper bound of $S$. 
	A \emph{chain complete poset} (cpo)  is a poset $C$ such that for every chain $S\subseteq C$, \ie a totally ordered subset, $\lub S$ exists. 
	A \emph{complete join semilattice} is a poset $J$ such that for any subset $S\subseteq J$,  $\lub S$ exists. 
	A  \emph{complete lattice} is a poset $L$ such that for any subset $S\subseteq L$, both $\glb S$ and $\lub S$ exist. 
	A function $f\colon P_1\to P_2$ between posets is \emph{monotone} if for all $x,y\in P_1$ such that $x\leq_{P_1}y$, it holds that $f(x)\leq_{P_2}f(y)$. 
	We refer to functions $O\colon C\to C$ with domain equal to the codomain as \emph{operators}. An element $x\in C$ is a \emph{fixpoint}
	of $O$ if $O(x)=x$.
	By Tarski's least fixpoint theorem, every monotone operator $O$ on a cpo has a least fixpoint, denoted $\lfp(O)$. To use a similar principle for operators stemming from non-monotonic logics, standard AFT \citep{DMT00ApproximationsStableOperatorsWell-FoundedFixpointsApplications} considers, for each complete lattice $\mathcal{L}$, its associated square bilattice $\langle L^2, \leqp\rangle$, where $\leqp$ is the \emph{precision order} on the Cartesian product $L^2$, \ie  $(x_1,y_1)\leqp (x_2,y_2)$ iff $x_1\leq_L x_2$ and $y_2\leq_L y_1$. 
	A square bilattice $\langle L^2, \leqp\rangle$ can be viewed as an approximation of $L$: an element $(x,y)\in L^2$ such that $x\leq_L y$ ``approximates'' all the values $z\in L$ such that $x\leq_L z\leq_L y$. Such pairs $(x,y)$ with $x\leq_L y$ are called \emph{consistent}. Pairs of the form $(x,x)\in L^2$ are called \emph{exact}, since they approximate only one element of $L$. 
	
	An \emph{approximator} $A\colon L^2\to L^2$ is a monotone operator that is \emph{symmetric}, \ie for all $(x,y)\in L^2$ it holds that $A_1(x,y)=A_2(y,x)$, where $A_1,A_2\colon L^2 \to L$ are the components of $A$, i.e.\ $A(x,y)=(A_1(x,y),A_2(x,y))$. 
	An approximator $A\colon L^2 \to L^2$ \emph{approximates} an operator $O\colon L\to L$ if for all $x\in L$, $A(x,x)=(O(x),O(x))$. Since $A$ is by definition monotone, by Tarski's theorem $A$ has a least fixpoint, which is called the \emph{Kripke-Kleene} fixpoint. Moreover, given an approximator $A$, there are three other operators which deserve our attention, together with their fixpoints: the operator approximated by $A$,
	$\subscript{O}{A}\colon x\in L\mapsto A_1(x,x)\in L$ whose fixpoints are called \emph{supported}; the \emph{stable operator} $\subscript{S}{A}\colon x\in L\mapsto \lfp(A_1(\cdot , x))\in L$ with the \emph{stable} fixpoints (where $A_1(\cdot, x)\colon y\in L\mapsto A_1(y,x)\in L$); and the \emph{well-founded operator} $\subscript{\mathcal{S}}{A}\colon (x,y)\in L^2\mapsto (\subscript{S}{A}(y), \subscript{S}{A}(x))\in L^2$, whose least fixpoint is referred to as the \emph{well-founded} fixpoint.  If $A$ is the four-valued immediate consequence operator \citep{F02Fixpointsemanticslogicprogrammingsurvey}, then the aformentioned four types of fixpoint correspond to the homonymous semantics of logic programing \citep{DMT00ApproximationsStableOperatorsWell-FoundedFixpointsApplications,DBV12ApproximationFixpointTheorySemanticsLogicAnswers}.
	
	The concepts presented so far are part of what we refer to as \emph{standard AFT}  \citep{DMT00ApproximationsStableOperatorsWell-FoundedFixpointsApplications}, i.e.\ the first account of AFT. Following this initial take, several other variants have been developed:  consistent AFT \citep{DMT03Uniformsemantictreatmentdefaultautoepistemiclogics}, non-deterministic AFT \citep{HAB24Non-deterministicapproximationfixpointtheoryapplicationdisjunctive}, or other extensions \citep{CRS18ApproximationFixpointTheoryWell-FoundedSemanticsHigher-Order}. In particular, the latter already proposes a way to deal with higher-order logic programs via an extension of consistent AFT. However, as already highlighted by \citet{BCCKPR24StableModelSemanticsHigher-OrderLogicProgramming}, the work of \citet{CRS18ApproximationFixpointTheoryWell-FoundedSemanticsHigher-Order} had some hidden problematic features. They can be summarised as follows:
	
	\begin{enumerate}
		\item \emph{The Approximator:} the well-founded semantics obtained via the approximator defined by \citet{CRS18ApproximationFixpointTheoryWell-FoundedSemanticsHigher-Order} does not  behave as expected when an existential quantifier occurs in the body of a rule. Take for instance the logic program with just the simple rule $\mathsf{p} \leftarrow \mathsf{R} \; \wedge \sim \mathsf{R}$, where $\mathsf{R}$ is a variable ranging over the booleans $\{\lfalse, \ltrue\}$. If we naively ground such program, we obtain $\mathsf{p}\leftarrow \lfalse \wedge \ltrue$ and $\mathsf{p}\leftarrow \ltrue \wedge \lfalse$, and $\mathsf{p}$ would clearly be evaluated as false. However, the approach adopted by \citet{CRS18ApproximationFixpointTheoryWell-FoundedSemanticsHigher-Order} uses \emph{approximated} elements. In particular, variables of type boolean range over $\{\lfalse, \ltrue,\mathbf{u}\}$. In more detail, for the logic program  $\mathsf{p} \leftarrow \mathsf{R}\; \wedge \sim \mathsf{R}$,  the approximator assignes to $\mathsf{p}$ the least upper bound of the body $\mathsf{R}\; \wedge \sim \mathsf{R}$, with $\mathsf{R}$ ranging over $\{\lfalse, \ltrue,\mathbf{u}\}$. Since such least upper bound is computed with respect to the \emph{truth order} $\lfalse \leq \mathbf{u}\leq \ltrue$, the predicate $\mathsf{p}$ is assigned the value $\lub\{\ltrue\wedge\lfalse, \lfalse\wedge\ltrue, \mathbf{u}\wedge\mathbf{u}\}=\mathbf{u}$ under the well-founded semantics. This contradicts the more intuitive and standard two-valued approach via grounding, which assigns $\lfalse$ to $\mathsf{p}$. In a way, allowing the existentially quantified variable to vary over all approximated elements seems to unnecessarily increase (w.r.t.\ the truth order) the value of the defined predicate, when evaluated under the well-founded semantics.
		\item \emph{The Notion of Exactness:} the work of \citet{CRS18ApproximationFixpointTheoryWell-FoundedSemanticsHigher-Order} lacks the notion of exactness for higher-order objects, which is fundamental in the context of AFT and rather non-trivial in the higher-order setting: exactness allows to recognize whether an approximated object, i.e.\ a pair $(x,y)$ in the bilattice, represents just one \emph{real} element of the lattice, and, in particular, when a model is two-valued. In other words, having such concept makes it possible to study not just the well-funded models, but also the stable ones.
	\end{enumerate}

In Section \ref{sec:iclp18}, we will show how we can use our novel concepts in the framework of \citet{CRS18ApproximationFixpointTheoryWell-FoundedSemanticsHigher-Order} to solve the issues listed above.
	
	%
	
	\subsection{Category Theory}\label{sec:cattheory}
	

	Category Theory (CT) studies mathematical structures and the relations between them, through the notion of a \emph{category}. Intuitively, a category $\catname{C}$ consists of a collection $\Obj{\catname{C}}$ of \emph{objects} and a collection $\Hom{\catname{C}}$ of relations, called \emph{morphisms}, between objects, satisfying some basic properties: every morphism $f$ has a \emph{domain} $s(f)$ and a \emph{codomain} $t(f)$, morphisms can be composed, and so on.
	
		\begin{definition}\label{def:category}
		A \emph{category} $\catname{C}$ consists of
		\begin{itemize}
			\item a collection of \emph{objects} $\Obj{\catname{C}}$,
			\item a collection of \emph{morphisms} $\Hom{\catname{C}}$,
			\item for every morphism $f\in\Hom{\catname{C}}$, an object $s(f)$ called the \emph{source} (or \emph{domain}) of $f$, and an object $t(f)$ called the \emph{target} (or \emph{codomain}) of $f$,
			\item for every object $X\in \Obj{\catname{C}}$, a morphism  $\mathit{id}_X$ called the \emph{identity morphism},
			\item for every two morphisms $f,g\in \Hom{\catname{C}}$ with $t(f)=s(g)$, a morphism $g\circ f$, called their \emph{composite},
		\end{itemize}
		such that 
		\begin{itemize}
			\item for all $f,g\in \Hom{\catname{C}}$ such that $t(f)=s(g)$, $s(g\circ f)=s(f)$
			\item for all $f,g\in \Hom{\catname{C}}$ such that $t(f)=s(g)$, $t(g\circ f)=t(g)$,
			\item for all $X\in \Obj{\catname{C}}$, $s(\mathit{id}_X)=t(\mathit{id}_X)=X$,
			\item for all $f,g,h\in\Hom{\catname{C}}$ such that $t(f)=s(g)$ and $t(g)=s(h)$, $(h\circ g)\circ f=h\circ(g\circ f)$,
			\item for all $X,Y\in \Obj{\catname{C}}$ and for all $f\in \Hom{\catname{C}}$ such that $s(f)=X$ and $t(f)=Y$, $f\circ \mathit{id}_X=f$ and $\mathit{id}_Y\circ f=f$.
		\end{itemize}
	\end{definition}

	In this paper, objects will always be certain ordered sets, and morphisms will be monotone functions. 	In the same way as morphisms between objects encode relations within a category, a morphism of categories, called a \emph{functor}, describes the relation between two categories.
	
	\begin{definition}\label{def:functor}
		Let $\catname{C},\catname{D}$ be two categories. A \emph{functor} $F\colon \catname{C}\to \catname{D}$ consists of a function
		$F_0\colon \Obj{\catname{C}}\to \Obj{\catname{D}}$ between the classes of objects, and a function $F_1\colon \Hom{\catname{C}}\to \Hom{\catname{D}}$, such that it respects target and source of morphisms, identity morphisms, and composition.
	\end{definition}
	
	For each $x,y\in \Obj{\catname{C}}$, we denote by $\hom_\catname{C}(x,y)$ the set of morphisms of $\catname{C}$ with domain $x$ and codomain $y$.
	
	\begin{definition}\label{def:full_faithful_embedding_subcat}
		A functor $F\colon \catname{C}\to \catname{D}$ is
		\begin{itemize}
			\item \emph{full} if for each $X,Y\in \Obj{\catname{C}}$, the map $F_1\!\!\restriction_{\hom_\catname{C}(X,Y)}\colon \hom_\catname{C}(X,Y) \to \hom_\catname{D}(F_0(X),F_0(Y))$ is surjective,
			\item \emph{faithful} if for each $X,Y\in \Obj{\catname{C}}$, the map $F_1\!\!\restriction_{\hom_\catname{C}(X,Y)}\colon \hom_\catname{C}(X,Y) \to \hom_\catname{D}(F_0(X),F_0(Y))$ is injective,
			\item \emph{embedding} if $F$ is faithful and $F_0$ is injective.
		\end{itemize}
		The domain of a full embedding $F\colon \catname{C}\to \catname{D}$ is called a \emph{full subcategory} of the codomain (denoted as $ \catname{C}\subseteq \catname{D}$).
	\end{definition}

	
	It is easy to see that we can define a category $\poset$ with objects the posets, and as morphisms the monotone functions between posets. We denote by $\CPO$, $\compjoinsl$, and $\CLat$ the full subcategories of $\poset$ with objects the cpo's, the complete join semilattices, and the complete lattices, respectively. Clearly, it also holds that $\CLat\subseteq\CPO\subseteq\poset$ and $\CLat\subseteq\compjoinsl\subseteq\poset$.
	
We are interested in inductively building \emph{approximation spaces} for higher-order concepts starting from \emph{base} ones. To be able to perform this construction, we need the approximation spaces to belong to a \emph{Cartesian closed} category, \ie a category $\catname{C}$ with a \emph{terminal object}, \emph{products}, and \emph{exponentials}. 


	\begin{definition}\label{def:terminal}
	$T\in\Obj{\catname{C}}$ is \emph{terminal} if for each $A\in \Obj{\catname{C}}$ there exists a unique morphism $f\colon A\to T$.
\end{definition}

	For instance, the poset with one element and trivial order is the terminal object of $\poset$, $\CPO$, $\compjoinsl$, and $\CLat$. 

\begin{definition}\label{def:product}
	Let $A_1, A_2\in \Obj{\catname{C}}$. A \emph{product} of $A_1$ and $A_2$ is an object of $\catname{C}$, denoted by $A_1\times A_2$, equipped with two morphisms $\pi_1\colon A_1\times A_2\to A_1$ and $\pi_2\colon A_1\times A_2\to A_2$, called \emph{first}, and \emph{second projection} respectively, such that for any $B\in \Obj{\catname{C}}$ and any morphisms $f_1\colon B\to A_1$ and $f_2\colon B\to A_2$, there exists a unique morphism $f_1\times f_2\colon B\to A_1\times A_2$ such that $\pi_1\circ (f_1\times f_2)=f_1$ and $\pi_2\circ (f_1\times f_2)=f_2$.
\end{definition}

	In $\poset$, and analogously for $\CPO$, $\compjoinsl$, and $\CLat$, the product of two objects $\mathcal{P}_1$ and $\mathcal{P}_2$ is the  Cartesian product of $P_1$ and $P_2$ equipped with the \emph{product order}, \ie $(x_1,y_1)\leqpr (x_2, y_2)$ if and only if $x_1\leq_{P_1} x_2$ and $y_1\leq_{P_2} y_2$.  The projections $\pi_1$ and $\pi_2$ are given by the usual Cartesian projections.

\begin{definition}\label{def:exponential}
	Let $A_1, A_2\in \Obj{\catname{C}}$. An \emph{exponential} of $A_1$ and $A_2$ is an object of $\catname{C}$, denoted by $A_2^{A_1}$, equipped with a morphism $\mathit{ev}\colon A_2^{A_1}\times A_1\to A_2$, called the \emph{evaluation}, such that for any $B\in \Obj{\catname{C}}$ and any morphism $f\colon B\times A_1\to A_2$, there exists a unique morphism $f'\colon B\to A_2^{A_1}$ such that $\mathit{ev}\circ(f'\times \mathit{id})=f$. 
\end{definition}

	In $\poset$, and analogously for $\CPO$, $\compjoinsl$, and $\CLat$, the exponential $\mathcal{P}_2^{\mathcal{P}_1}$  is the set of monotone functions from $P_1$ to $P_2$ equipped with  the \emph{pointwise order} (induced by $\leq_{P_2}$), \ie $f_1\leqpt f_2$ if and only if for all $x\in P_1$, $f_1(x)\leq_{P_2} f_2(x)$.  The evaluation $\mathit{ev}$ is given by the usual function evaluation, \ie $\mathit{ev}(f,x)=f(x)$.

\begin{definition}\label{def:Cartesian_closed}
	A category \catname{C} is \emph{Cartesian closed} if it has a terminal object, and for each $A_1, A_2\in\Obj{\catname{C}}$, there exist $A_1\times A_2\in\Obj{\catname{C}}$ and $A_2^{A_1}\in\Obj{\catname{C}}$. 
\end{definition}

By our prior observations, it follows that $\poset, \CPO$, $\compjoinsl$, and $ \CLat$ are all Cartesian closed.

	In the context of AFT, we are often interested in the space of interpretations over a possibly infinite vocabulary	of a logic program. In order 
to include this, we need another notion from CT, namely a generalized version of the categorical product for (possibly infinite) families of objects. 

\begin{definition}\label{def:generalized_product}
	Let $\{A_i\}_{i\in I}$ be a family of objects of a category $\catname{C}$ indexed by $I$. The \emph{(generalized) product} of the family $\{A_i\}_{i\in I}$ is an object of $\catname{C}$, denoted by $\Pi_{i\in I} A_i$, equipped with morphisms $\pi_i\colon \Pi_{i\in I} A_i\to A_i$, called the \emph{$i$-th projection}, such that for all $B \in\Obj{\catname{C}}$ and for all families of  morphisms  $\{\varphi_i\colon B\to A_i\}_{i\in I}$ indexed by $I$, there exists a unique $\psi\colon B\to  \Pi_{i\in I} A_i $ such that for all $i\in I$ it holds that $\varphi_i=\pi_i\circ \psi$.
\end{definition}

We say that a category $\catname{C}$ \emph{has generalized products} if  for all families $\{A_i\}_{i\in I}$ of objects of $\catname{C}$ the product $\Pi_{i\in I} A_i\in \Obj{\catname{C}}$ exists. Clearly, if $\catname{C}$ is  Cartesian closed and the index $I$ is finite, this product always exists; but this is not always the case for infinite $I$. However, we will show in the remainder of this subsection that every full subcategory of $\poset$ has generalized products (Proposition \ref{prop:full_subcat_poset_genproducts}). We wish to warn the reader that the following category-theoretic notions are only meant to support the proofs of Propositions \ref{prop:full_subcat_poset_genproducts} and \ref{prop:iso_product_functions}, and will not be used in the next sections of the paper.

The generalized product is a special case of a very common construction in CT, called the \emph{limiting cone}, or simply \emph{limit}. We proceed with the definitions leading to the concept of limit.

\begin{definition}
	A \emph{diagram} $X_\bullet$ in a category $\catname{C}$ is
	\begin{enumerate}
		\item a set $\{X_i\}_{i\in I}$ of objects of $\catname{C}$,
		\item for every pair $(i,j)\in I\times I$, a set $\{f_\alpha\colon X_i \to X_j\}_{\alpha\in I_{i,j}}$ of morphisms,
		\item for every $i\in I$ an element $\epsilon_i \in I_{i,i}$,
		\item for each $(i,j,k)\in I\times I \times I$ a function $\mathit{comp}_{i,j,k}
		\colon I_{i,j}\times I_{j,k} \to I_{i,k}$ such that
		\begin{enumerate}
			\item $\mathit{comp}$ is associative and unital with the $f_{\epsilon_i}$'s being the neutral elements,
			\item for every $i\in I$, $f_{\epsilon_i}=\mathit{id}_{X_i}$ is the identity morphism of $X_i$,
			\item for every two composable morphisms $f_\alpha\colon X_i \to X_j$ and $f_\beta \colon X_j\to X_k$, it holds that $f_\beta \circ f_\alpha =f_{\mathit{comp}_{i,j,k}(\alpha, \beta)}$.
		\end{enumerate}
	\end{enumerate}
\end{definition}

\begin{definition}
	Let $X_\bullet=(\{f_\alpha\colon X_i\to X_j\}_{i,j\in I, \alpha\in I_{i,j}}, \mathit{comp})$ be a diagram in $\catname{C}$. A \emph{cone} over $X_\bullet$ is an object $X\in\Obj{\catname{C}}$ together with, for each $i\in I$, a morphism $p_i\colon X\to X_i\in\Hom{\catname{C}}$ such that for all $(i,j)\in I\times I$ and for all $\alpha\in I_{i,j}$, it holds that $f_\alpha \circ p_i =p_j$.
	\\	Moreover, the \emph{limiting cone} or \emph{limit} of $X_\bullet$ is, if it exists, the cone over $X_\bullet$ which is \emph{universal} among all possible cones over $X_\bullet$.
\end{definition}

\begin{definition}
	A functor $F\colon \catname{C}\to \catname{D}$ \emph{reflects all limits} if for all diagrams $X_\bullet$ in $\catname{C}$, and for all cones $C$ over $X_\bullet$ such that $F(C)$ is a limiting cone of $F(X_\bullet)$, $C$ is a limiting cone of $X_\bullet$. 
\end{definition}

\begin{proposition}\label{prop:reflects_limits}
	A full and faithful functor reflects all limits.
\end{proposition}
\begin{proof}
	Lemma 3.3.5 in \citep{R17Categorytheorycontext}. 
\end{proof}

We are finally able to prove that every full subcategory of $\poset$ has geeralized products.

\begin{proposition}\label{prop:full_subcat_poset_genproducts}
	If $\catname{C}\subseteq\poset$, then $\catname{C}$ has generalized products.
\end{proposition}
	\begin{proofwithoutqed}
	By Definition \ref{def:full_faithful_embedding_subcat}, there exists a full embedding $F\colon \catname{C}\to \poset$. By Proposition \ref{prop:reflects_limits}, $F$ reflects all limits. Moreover, it is clear that $\poset$ has generalized products. Since generalized products are a type of limit, we conclude that $\catname{C}$ has generalized products.\qed
\end{proofwithoutqed}

In a full subcategory of $\poset$ we can rewrite a generalized product as a poset of functions.
	Given two posets $\mathcal{P}_1, \mathcal{P}_2\in\Obj{\poset}$, we denote by $(\mathcal{P}_1\to\mathcal{P}_2)\in\Obj{\poset}$ the poset of functions from $\mathcal{P}_1$ to $\mathcal{P}_2$ ordered with the pointwise order. In particular, notice that $(\mathcal{P}_1\to\mathcal{P}_2)$ may contain non-monotone functions.
	
\begin{proposition}\label{prop:iso_product_functions}
	Let $\catname{C}$ be a full subcategory of $\poset$, $X\in\Obj{\poset}$, and $Y\in \Obj{\catname{C}}$. Then there exists an isomorphism $(X\to Y)\cong \Pi_{x\in X} Y$ in $\catname{C}$. 
\end{proposition}
	\begin{proof}
	By Proposition \ref{prop:full_subcat_poset_genproducts}, $ \Pi_{x\in X} Y \in \Obj{\catname{C}}$. Moreover, there exists a bijection $\varphi\colon (X\to Y)\to \Pi_{x\in X} Y$, defined by $\varphi(f)=(f(a))_{a\in X}$, with inverse $\varphi^{-1}\colon  \Pi_{x\in X} Y \to (X\to Y)$, defined by, for all $a\in X$, $\varphi^{-1}((y_x)_{x\in X})(a)=y_a$. Because of the definition of pointwise order and product order, it is immediate to show that both $\varphi$ and $\varphi^{-1}$ preserve the orders, \ie they are monotone.
\end{proof}
	
	%

	\section{The Approximation System}\label{sec:categoryApprox}
	
	In this section, we introduce the notions of \emph{approximation category} and of \emph{approximation system}, which constitute the core of the theoretical framework for AFT we developed. 
	
	Let $\mathcal{L}$ be a higher-order language based on a hierarchy of types
	$\mathbb{H}$
	comprising of \emph{base types} $\tau$, and two kinds of composite types: \emph{product types} $\Pi_{i\in I} \tau_i$, and \emph{morphism types} $\tau_1\to\tau_2$. For instance, a base type could be the boolean type $o$ or the type $\iota$ of individuals, whereas in the composite types we may find the type $\iota\to o$, which is the type of unary first-order predicates
	. We denote by $\mathcal{B}_\mathbb{H}$ 
	the set of base types. For the sake of simplicity, we omit the subscript of $\mathcal{B}$ when it is clear from the context of use.

		We associate to each type $\tau$ of $\mathcal{B}_\mathbb{H}$,  an object $E_\tau\in\Obj{\poset}$, and we define inductively for all $\{\tau_i\}_{i\in I}\subseteq \mathbb{H}$, $E_{\Pi_{i\in I}\tau_i}=\Pi_{i\in I}E_{\tau_i}$, and   for all $\tau_1,  \tau_2\in \mathbb{H}$, $E_{\tau_1\to\tau_2}=(E_{\tau_1}\to E_{\tau_2})$.
		The object $E_\tau$ is called the \emph{semantics of $\tau$}.
		For example, if the semantics of the boolean type $o$ is chosen to be $E_o:=\{\lfalse, \ltrue\}$ with the standard truth ordering, then the semantics for type $o\to o$ is the poset of functions from $E_o$ to $E_o$. 
		
		In many applications of AFT, we are ultimately interested in the space of interpretations, which associate to each symbol of a vocabulary, an element of the semantics of the type of such symbol. It follows that an interpretation can be seen as a tuple of elements of different semantics. In more detail, given a vocabulary $V$, we can consider the product type $\tau=\Pi_{s\in V} t(s)$, where $t(s)$ is the type of the symbol $s$. Then, the space of interpretations for the vocabulary $V$ coincides with the semantics $E_\tau=\Pi_{s\in V} E_{t(s)}$. 
		
		We have so far defined the semantics of all the base types and the composite ones constructed from them. 
		Notice that, it is often not necessary to define the spaces of approximations for all such semantics $E_\tau$, which are infinitely many.
		Because of the nature of our formalism, we can easily restrict the set of types we take into account: we can fix a subset $\mathbb{T}\subseteq \mathbb{H}$ of types, and focus our attention onto the set $S_\mathbb{T}$ defined as follows:
		\begin{itemize}
			\item for all $\tau \in \mathbb{T}$, $E_\tau\in S_\mathbb{T}$,
			\item if $E_{\tau_1\to\tau_2}\in S_\mathbb{T}$, then $E_{\tau_2}\in S_\mathbb{T}$, 
			\item if $E_{\Pi_{i\in I}\tau_i} \in S_\mathbb{T} $, then $E_{\tau_i}\in S_\mathbb{T}$ for all $i\in I$.
		\end{itemize} 
		We will dive deeper into this matter in Section  \ref{sec:iclp18} 
		 where we present applications of our framework. We denote by $\mathcal{B}_\mathbb{T}$ 
		the set of base types of $\mathbb{H}$ belonging to $\mathbb{T}$. 

		The notion of \emph{approximation system} (Definition \ref{def:approximation_system}) together with what follows in this section, provide a general framework in which the techniques of AFT can be applied on higher-order languages. Before stating the, rather lengthy, definition of an approximation system, we provide an intuitive understanding of its components.
		
		For each $\subscript{E}{\famtypes}\in S_\mathbb{T}$, we shall consider a corresponding space $\mathit{App}(\subscript{E}{\famtypes})$, called an \emph{approximation space}, whose elements approximate the elements of $\subscript{E}{\famtypes}$. Hence, we define a Cartesian closed full subcategory of $\CPO$, denoted by $\Approx$, and a map $\mathit{App}\colon S_\mathbb{T}\to\Obj{\Approx}$ encoding such correspondence. The fact that $\Approx\subseteq\CPO$ allows us to apply the Knaster-Tarski theorem on the approximation spaces, and guarantees the existence of generalized products (Proposition \ref{prop:full_subcat_poset_genproducts}).  
		Notice that, even though we fixed a mapping $\mathit{App}$ between the set $S_\mathbb{T}$ and the objects of $\Approx$, there is, so far, no relation between the elements of $\subscript{E}{\famtypes}$ and those of $\mathit{App}(\subscript{E}{\famtypes})$.  
		The approximation space $\mathit{App}(\subscript{E}{\famtypes})$ is meant to approximate the elements of $\subscript{E}{\famtypes}$. In particular, we want the order $\leq_{\mathit{App}(\subscript{E}{\famtypes})}$ on $\mathit{App}(\subscript{E}{\famtypes})$, which we call a \emph{precision order}, to encode the approximating nature of $\mathit{App}(\subscript{E}{\famtypes})$ for $\subscript{E}{\famtypes}$: intuitively,  $a\leq_{\mathit{App}(\subscript{E}{\famtypes})} b$ if $a$ is \emph{less precise} than $b$, \ie if an element $e\in \subscript{E}{\famtypes}$ is approximated by $b$, then $e$ is also approximated by $a$. In the context of AFT, of particular interest are the elements of $\mathit{App}(\subscript{E}{\famtypes})$ which approximate just one element, called the \emph{exact} elements. Thus, in the definition of approximation system that we are about to give, for every base type $\tau\in\mathcal{B}_\mathbb{T}$, we fix a set $\mathcal{E}_\tau$ of exact elements of $\mathit{App}(E_\tau)$, and a function $\zeroproj{\tau} \colon \mathcal{E}_\tau\to E_\tau$, which associates each exact element to the unique element of $E_\tau$ it represents. 
		
		To obtain a sensible framework, it is fundamental to carefully define the sets of exact elements and a projection that associates each exact element to the object it represents. Hence, we impose conditions on the possible choices of the sets $\mathcal{E}_\tau$ and the functions  $\zeroproj{\tau}$, for $\tau\in\mathcal{B}_\mathbb{T}$. 
		Since an exact element of $\mathit{App}(E_\tau)$ approximates a single element of the semantics $E_\tau$, if both $a$ and $b$ are exact and one is more precise than the other, then they should represent the same element, i.e.\  $\zeroproj{\tau}(a)=\zeroproj{\tau}(b)$ (Item \ref{item:condition_comparable_exact} in Definition \ref{def:approximation_system}). This requirement also hints at a very important fact: the definition of approximation system allows for the existence of multiple exact elements of $\mathit{App}(E_\tau)$ representing the \emph{same} element of $E_\tau$. Because of this possible multitude of exact representatives, we want to have, for each element $e\in E_\tau$, a natural choice for a representative in the approximation space $\mathit{App}(E_\tau)$. This is why, for each element $e\in E_\tau$, we require that the greatest lower bound of all the exact elements representing $e$ exists, is exact, and represents $e$ (Item \ref{item:condition_glb_exact} in Definition \ref{def:approximation_system}). 
		Lastly, we add one more condition on exact elements to accommodate several existing versions of AFT. In consistent AFT \citep{DMT03Uniformsemantictreatmentdefaultautoepistemiclogics}, exact elements are maximal, while in standard AFT, this is not the case, and there are elements beyond exact ones. We require that either the exact elements are maximal, or we can take arbitrary joins in the approximation spaces (Item~\ref{item:condition_completejoinsemilattice} in Definition \ref{def:approximation_system}). This last condition will later allow for a generalization of both $\mathcal{E}_\tau$ and $\zeroproj{\tau}$ to any type $\tau$ of $\mathbb{H}$, satisfying properties analogous to the ones required for the base types counterparts (Propositions \ref{prop:consistent_aft} and \ref{prop:proj_well-def}).
		
		We are now ready to state the definition of an approximation system. We write~$f^{-1}(b)$  for the \emph{preimage} of an element $b\in B$ via a function $f\colon A\to B$, \ie  $f^{-1}(b)=\{a\mid f(a)=b \}\subseteq A$. Recall that given two posets $\mathcal{P}_1, \mathcal{P}_2\in\Obj{\poset}$, we denote by $(\mathcal{P}_1\to\mathcal{P}_2)\in\Obj{\poset}$ the poset of (possibly non-monotone) functions from $\mathcal{P}_1$ to $\mathcal{P}_2$ ordered with the pointwise order.
		
		\begin{definition}\label{def:approximation_system}
			A tuple $(\Approx, \mathit{App}, \{\mathcal{E}_\tau\}_{\tau \in\mathcal{B}}, \{\zeroproj{\tau}\}_{\tau\in \mathcal{B}})$ is an  \emph{approximation system (for $S_\mathbb{T}$)} if
			\begin{enumerate}
				\item $\Approx$ is a Cartesian closed full subcategory of $\CPO$, called the \emph{approximation category}. The objects of $\Approx$ are called \emph{approximation spaces}. 
				\item $\mathit{App}\colon S_\mathbb{T}\to\Obj{\Approx}$ is a function such that for all $E_\tau \in S_\mathbb{T}$
				\begin{enumerate}
					\item  if  $\tau=\Pi_{i\in I}\tau_i$ is a product type, then $\mathit{App}(E_\tau)=\Pi_{i\in I}\mathit{App}(E_i)$,
					\item\label{item:app_function}  if  $\tau=\tau_1\to\tau_2$ and $E_{\tau_1}\notin S_\mathbb{T}$, then $\mathit{App}(E_{\tau_1 \to\tau_2})=(E_{\tau_1}\to \mathit{App}(E_{\tau_2}))$,
					\item\label{item:app_morphism} if $\tau=\tau_1\to\tau_2$ and $E_{\tau_1}\in S_\mathbb{T}$, then $\mathit{App}(E_{\tau_1 \to\tau_2})=\mathit{App}(E_{\tau_2})^{\mathit{App}(E_{\tau_1})}$.
					
				\end{enumerate} 
				\item $\{\mathcal{E}_\tau\}_{\tau \in\mathcal{B}}$ is a family of sets such that the following hold:
				\begin{enumerate}
					\item for each base type $\tau\in \mathcal{B}$, $\mathcal{E}_\tau\subseteq \mathit{App}(E_\tau)$,
					\item\label{item:condition_completejoinsemilattice} either  $\mathit{App}(E_\tau)\in \Obj{\compjoinsl}$ for all $\tau\in\mathcal{B}$, or for all $\tau\in\mathcal{B}$, if $a\in \mathcal{E}_\tau$ and $b\in\mathit{App}(E_\tau)$ such that $a\leq_{\mathit{App}(E_\tau)} b$, then also $b\in\mathcal{E}_\tau$.
				\end{enumerate} 
				\item $\{\zeroproj{\tau}\}_{\tau\in \mathcal{B}}$ is a family of surjective functions such that 
				for each base type $\tau\in \mathcal{B}$: 
				\begin{enumerate}
					\item $\zeroproj{\tau} \colon \mathcal{E}_\tau\to E_\tau$,
					\item\label{item:condition_comparable_exact} for all $a,b\in \mathcal{E}_\tau$,  if $a\leq_{\mathit{App}(E_\tau)} b$, then $\zeroproj{\tau}(a)=\zeroproj{\tau}(b)$,
					\item\label{item:condition_glb_exact} for all $e\in E_\tau$, there exists $\glb((\zeroproj{\tau})^{-1}(e))\in \mathcal{E}_\tau$ and $\zeroproj{\tau}(\glb(\zeroproj{\tau})^{-1}(e))=e$.
				\end{enumerate} 
				
			\end{enumerate}
		\end{definition}
		Notice that, by Proposition \ref{prop:iso_product_functions}, the object $(E_{\tau_1}\to \mathit{App}(E_{\tau_2}))$ in Item \ref{item:app_function} of Definition \ref{def:approximation_system} is indeed an object of the approximation category $\Approx$.
		Morover, again by Proposition \ref{prop:iso_product_functions}, it holds that $E_{\tau_1\to\tau_2}=(E_{\tau_1}\to E_{\tau_2})\cong \Pi_{i\in E_{\tau_1}}E_{\tau_2}=E_{\Pi_{i\in E_{\tau_1}}\tau_2}$. However, in Item \ref{item:app_morphism} of the above definition, we have $\mathit{App}(E_{\tau_1 \to\tau_2})=\mathit{App}(E_{\tau_2})^{\mathit{App}(E_{\tau_1})}\not\cong \Pi_{i\in E_{\tau_1}}\mathit{App}(E_{\tau_2})=\mathit{App}(E_{\Pi_{i\in E_{\tau_1}}\tau_2})$. Hence,
		while the map $\mathit{App}$, in a way, respects the structure given by the type hierarchy $\mathbb{H}$, it does not commute with isomorphisms of posets.
		
		Finally, it is important to notice that, while the approximation system depends on the application at hand, \ie on the language, the semantics, and so on, the approximation category depends only on the version of AFT. 
		
		We now fix an approximation system $\mathcal{S}=(\Approx, \mathit{App}, \{\mathcal{E}_\tau\}_{\tau \in\mathcal{B}}, \{\zeroproj{\tau}\}_{\tau\in \mathcal{B}})$ for $S_\mathbb{T}$, and extend the notion of exactness to all approximation spaces.
		
		\begin{definition}\label{def:exact}
			Let $\subscript{E}{\famtypes}\in S_\mathbb{T}$. An element $e\in \mathit{App}(\subscript{E}{\famtypes})$ is \emph{exact} if one of the following conditions holds:
			\begin{enumerate}
				\item $\tau\in\mathcal{B}_\mathbb{T}$ and $e \in \mathcal{E}_\tau$,
				\item $\famtypes= \Pi_{i\in I}\tau_i$ and for each $i\in I$, the $i$-th component $\pi_i(e)$ of $e$ is exact,
				\item $\tau=\tau_1\to\tau_2$,  $E_{\tau_1}\notin S_\mathbb{T}$, and 
				for all $e_1\in E_{\tau_1}$, $e(e_1)\in \mathit{App}(E_{\tau_2})$ is exact.
				\item\label{item:exact_morphisms} $\tau=\tau_1\to\tau_2$,  $E_{\tau_1}\in S_\mathbb{T}$, and for all $e_1\in \mathit{App}(E_{\tau_1})$ exact, $e(e_1)\in \mathit{App}(E_{\tau_2})$ is exact.
			\end{enumerate}
		\end{definition}		
		
	The reader may wonder why  there are two cases for a morphism type $\tau_1\to\tau_2$ in Definition \ref{def:exact}, depending whether $E_{\tau_1}$ is in $S_\mathbb{T}$ or not, i.e.\ whether we approximate the elements of the semantics of $\famtypes_1$ or not. Recall that an exact element in the approximation space is meant to represent one and only one element of the semantics of the same type. Intuitively, for a morphism type $\tau_1\to\tau_2$, a function in $\mathit{App}(E_{\tau_1\to\tau_2})$ is exact when the image of any exact element is exact. This is indeed sufficient and we do not need to consider the image of non-exact elements of the domain: we will prove in Proposition \ref{prop:proj_well-def} that there is an exact element in the approximation space for each element of the semantics of the corresponding type. Now, considering \emph{exact elements of the domain} (and looking at their image) makes sense only when the domain is an approximation space, i.e.\ when $E_{\tau_1}\in S_\mathbb{T}$; in the other case, when $E_{\tau_1}\notin S_\mathbb{T}$ we can directly consider the elements of the semantics, as they do not get approximated. 


	For $\tau \in \mathbb{T}$, we denote by $\mathcal{E}_\tau$ the subset of $\mathit{App}(E_\tau)$ of exact elements of type $\tau$.  
		The following proposition shows that the condition \ref{item:condition_completejoinsemilattice} of Definition \ref{def:approximation_system} holds for any $\subscript{E}{\famtypes}\in S_\mathbb{T}$.
		
	
		\begin{proposition}\label{prop:consistent_aft}
		Either for all $\subscript{E}{\famtypes}\in S_\mathbb{T}$ it holds that $\mathit{App}(E_\famtypes)\in \Obj{\compjoinsl}$, or  for all $\subscript{E}{\famtypes}\in S_\mathbb{T}$, for all $b\in\mathit{App}(\subscript{E}{\famtypes})$, and for all $e\in\subscript{\mathcal{E}}{\famtypes}$, if $e\leq_{\mathit{App}(\subscript{E}{\famtypes})} b$ , then $b\in\subscript{\mathcal{E}}{\famtypes}$.
	\end{proposition}
	
		\begin{proof}
			
		 Suppose there exists $\subscript{E}{\famtypes'}\in S_\mathbb{T}$ such that $\mathit{App}(E_{\famtypes'})\notin \Obj{\compjoinsl}$. We have to show that for all $\subscript{E}{\famtypes}\in S_\mathbb{T}$, for all $b\in\mathit{App}(\subscript{E}{\famtypes})$, and for all $e\in\subscript{\mathcal{E}}{\famtypes}$, if $e\leq_{\mathit{App}(\subscript{E}{\famtypes})} b$ , then $b\in\subscript{\mathcal{E}}{\famtypes}$.  We proceed by induction on $\famtypes$.
		 
		 Let $\tau$ be a base type. Since we assumed that $\mathit{App}(E_{\famtypes'})\notin \Obj{\compjoinsl}$ for some $\subscript{E}{\famtypes'}\in S_\mathbb{T}$,
		  and $\compjoinsl$ is Cartesian closed, and has generalized products by Proposition \ref{prop:full_subcat_poset_genproducts}, then there must exists a $\sigma \in \mathcal{B}$ such that $\mathit{App}(E_\sigma)\notin \Obj{\compjoinsl}$. Thus, by condition \ref{item:condition_completejoinsemilattice} of Definition \ref{def:approximation_system} we can conclude the base step of the induction.

		Now let $\famtypes=\Pi_{i\in I}\tau_i$ and suppose the proposition hold for $E_{\tau_i}$ for all $i\in I$.  Let $(b_i)\in \mathit{App}(\subscript{E}{\famtypes})$ such that $e:=(e_i)\leq_{\mathit{App}(\subscript{E}{\famtypes})} (b_i)$. By the definition of the product order, Definition \ref{def:exact}, and the induction hypothesis, we get that $b_i\in\mathcal{E}_{\tau_i}$ for all $i\in I$, \ie $(b_i)\in \subscript{\mathcal{E}}{\famtypes}$, as desired. 
		
		Let $\famtypes=\tau_1\to \tau_2$ with $E_{\tau_1}\notin S_\mathbb{T}$, and suppose the proposition hold for $E_{\tau_2}$. By Proposition \ref{prop:iso_product_functions}, it holds that $\mathit{App}(E_\tau)\cong\mathit{App}(\Pi_{i\in E_{\tau_1}} E_{\tau_2})$, thus, we can reduce to the previous case.
		
		Let $\famtypes=\tau_1\to \tau_2$ with $E_{\tau_1}\in S_\mathbb{T}$, and suppose the proposition hold for $E_{\tau_1}$ and $E_{\tau_2}$. Let $f\in \mathit{App}(\subscript{E}{\famtypes})$ such that $e\leq_{\mathit{App}(\subscript{E}{\famtypes})} f$. For $f$ to be exact, it must send exact elements to exact elements. Let $a\in\mathcal{E}_{\tau_1}$. By the definition of the order on morphisms, and Defintion \ref{def:exact}, it holds that $f(a)\geq_{\mathit{App}(E_{\tau_2})} e(a)\in \mathcal{E}_{\tau_2}$. By induction hypothesis, it follows that $f(a)\in \mathcal{E}_{\tau_2}$. Hence, $f\in\subscript{\mathcal{E}}{\famtypes}$, as desired.
	\end{proof}

		Now that we have defined the exact elements 
		for any semantics in $S_\mathbb{T}$, we extend the family  $\{\zeroproj{\tau}\}_{\tau\in \mathcal{B}}$ to have a map for each $\subscript{E}{\famtypes}\in S_\mathbb{T}$.
		We can do this inductively, by defining a new family of functions $\{\mathfrak{p}_\famtypes\colon \subscript{\mathcal{E}}{\famtypes} \to \subscript{E}{\famtypes} \}_{\subscript{E}{\famtypes}\in S_\mathbb{T}}$ as follows:
		
		\begin{enumerate}
			\item if $\tau\in\mathcal{B}$, then $\mathfrak{p}_\famtypes:=\zeroproj{\tau}$,
			\item if $\famtypes=\Pi_{i\in I}\tau_i$, then for all $(e_i)_{i\in I}\in\subscript{\mathcal{E}}{\famtypes}$,  $\mathfrak{p}_\famtypes ((e_i)_{i\in I}):= (\mathfrak{p}_{\tau_i}(e_i))_{i\in I}$,
			\item if $\famtypes=\tau_1\to\tau_2$, and $E_{\tau_1}\notin S_\mathbb{T}$, then for all $f\in \subscript{\mathcal{E}}{\famtypes}$, and for all $e\in E_{\tau_1}$, $\mathfrak{p}_\famtypes(f)(e):=\mathfrak{p}_{\tau_2}(f(e))$.
			\item if $\famtypes=\tau_1\to\tau_2$, and $E_{\tau_1}\in S_\mathbb{T}$, then for all $f\in \subscript{\mathcal{E}}{\famtypes}$, and for all $e\in E_{\tau_1}$, $\mathfrak{p}_\famtypes(f)(e):=\mathfrak{p}_{\tau_2}(f(d))$, where $d\in\mathfrak{p}_{\tau_1}^{-1}(e)$, \ie $\mathfrak{p}_{\tau_1}(d)=e$.
		\end{enumerate}
		
		Recall that, intuitively, the function $\mathfrak{p}_\famtypes$ sends an exact element of type $\famtypes$ to the element it represents in the semantics of $\famtypes$.
		
		In the following proposition, we prove that for each $\subscript{E}{\famtypes}\in S_\mathbb{T}$, the function $\mathfrak{p}_\famtypes$ is well-defined, surjective, and satisfies properties analogous to \ref{item:condition_comparable_exact} and \ref{item:condition_glb_exact} of Definition \ref{def:approximation_system}.

						\begin{proposition}\label{prop:proj_well-def}
								Let $\subscript{E}{\famtypes}\in S_\mathbb{T}$, $e_1, e_2\in \subscript{\mathcal{E}}{\famtypes}$, and $e\in \subscript{E}{\famtypes}$. The following statements hold:
								
								\begin{enumerate}
										\item\label{item:well-def} $\mathfrak{p}_\famtypes$ is well-defined.
										\item\label{item:surjective} $\mathfrak{p}_\famtypes$ is surjective.
										\item\label{item:comparable_exact} if  $e_1\leq_{\mathit{App}(\subscript{E}{\famtypes})} e_2$, then $\mathfrak{p}_\famtypes(e_1)=\mathfrak{p}_\famtypes(e_2)$.
										\item\label{item:glb_of_exact} there exists $\glb\mathfrak{p}_\famtypes^{-1}(e)\in\subscript{\mathcal{E}}{\famtypes}$ and $\mathfrak{p}_\famtypes(\glb\mathfrak{p}_\famtypes^{-1}(e))=e$. 
									\end{enumerate}
							\end{proposition}

		\begin{proof}
		We proceed by induction on $\famtypes$.
		Let  $\tau\in \mathcal{B}_\mathbb{T}$. Then $\mathfrak{p}_\famtypes=\zeroproj{\tau}$ and the Items \ref{item:well-def}, \ref{item:surjective}, \ref{item:comparable_exact}, and \ref{item:glb_of_exact} hold by definition of $\zeroproj{\tau}$.
		
		Now suppose $\famtypes=\Pi_{i\in I}\tau_i$, and assume that Items \ref{item:well-def}, \ref{item:surjective}, \ref{item:comparable_exact}, and  \ref{item:glb_of_exact} hold for $\tau_i$ for all $i \in I$. Since $\mathfrak{p}_{\tau_i}$ is well-defined and surjective by hypothesis for all $i \in I$, it is clear by definition that also $\mathfrak{p}_\famtypes$ is well-defined and surjective. Let $(a_i), (b_i)\in \subscript{\mathcal{E}}{\famtypes}$ such that $(a_i)\leq_{\mathit{App}(\subscript{E}{\famtypes})} (b_i)$. By the product order on $\mathit{App}(\subscript{E}{\famtypes})$ we have $a_i \leq_{\mathit{App}(E_{\tau_i})} b_i$ for all $i\in I$. By definition \ref{def:exact}, $a_i, b_i \in\mathcal{E}_{\tau_i}$ for all $i\in I$.  Hence, by hypothesis it follows that $\mathfrak{p}_{\tau_i}(a_i)=\mathfrak{p}_{\tau_i}(b_i)$ for all $i\in I$. By definition of $\mathfrak{p}_\famtypes$, we get that $\mathfrak{p}_\famtypes((a_i))=\mathfrak{p}_\famtypes((b_i))$. Thus, Item  \ref{item:comparable_exact} hold. 
		Now let $(a_i)\in\ \subscript{E}{\famtypes}$. 
		By the definition of $\mathfrak{p}_\famtypes$, it is easy to see that $\mathfrak{p}_\famtypes^{-1}((a_i))=\Pi_{i\in I} \mathfrak{p}_{\tau_i}^{-1}(a_i)$. Hence, by induction hypothesis for Item \ref{item:glb_of_exact}, we have  that $\glb(\mathfrak{p}_\famtypes^{-1}((a_i)))=\glb(\Pi_{i\in I}
		\mathfrak{p}_{\tau_i}^{-1}(a_i))=\Pi_{i\in I} \glb\mathfrak{p}_{\tau_i}^{-1}(a_i)\in \Pi_{i\in I} \mathfrak{p}_{\tau_i}^{-1}(a_i)=\mathfrak{p}_\famtypes^{-1}((a_i))$, where the second equality holds because of the definition of the product order on $\mathit{App}(\subscript{E}{\famtypes})$. Thus, also Item  \ref{item:glb_of_exact} hold for $\mathfrak{p}_\famtypes$.
		
		
		Suppose $\famtypes=\tau_1\to \tau_2$ with $E_{\tau_1}\notin S_\mathbb{T}$, and assume that Items \ref{item:well-def}, \ref{item:surjective}, \ref{item:comparable_exact}, and \ref{item:glb_of_exact} hold for and $\tau_2$. By Proposition \ref{prop:iso_product_functions} and Definition \ref{def:approximation_system}, we have $E_{\tau_1\to \tau_2}=(E_{\tau_1}\to E_{\tau_2})\cong \Pi_{i\in E_{\tau_1}} E_{\tau_2}$ and $\mathit{App}(E_{\tau_1\to \tau_2})\cong\mathit{App}(\Pi_{i\in E_{\tau_1}} E_{\tau_2})$. It is easy to see that we can reduce to the previous case with $I:=E_{\tau_1}$.
		
		Finally, suppose $\famtypes=\tau_1\to \tau_2$ with $E_{\tau_1}\in S_\mathbb{T}$, and assume that Items \ref{item:well-def}, \ref{item:surjective}, \ref{item:comparable_exact}, and \ref{item:glb_of_exact} hold for $\tau_1$ and $\tau_2$. We first show that $\mathfrak{p}_\famtypes$ is well-defined, \ie for all $f\in \subscript{\mathcal{E}}{\famtypes}$ there exists unique $g\in \subscript{E}{\famtypes}=E_{\tau_2}^{E_{\tau_1}}$ such that $\mathfrak{p}_\famtypes(f)=g$. Let $f\in \subscript{\mathcal{E}}{\famtypes}$. First notice that, since $\mathfrak{p}_{\tau_1}$ is surjective by hypothesis, for all $e\in E_{\tau_1}$ there exists $d\in\mathcal{E}_{\tau_1}$ such that $\mathfrak{p}_{\tau_1}(d)=e$. Moreover, since $\mathfrak{p}_{\tau_2}$ is well-defined by hypothesis, for all $e\in E_{\tau_1}$ and $d\in \mathfrak{p}_{\tau_1}^{-1}(e)$, we get an element $\mathfrak{p}_{\tau_2}(f(d))\in E_{\tau_2}$. It remains to show that for all $e\in E_{\tau_1}$, if $d_1,d_2\in  \mathfrak{p}_{\tau_1}^{-1}(e)\subseteq \mathcal{E}_{\tau_1}$, then $\mathfrak{p}_{\tau_2}(f(d_1))=\mathfrak{p}_{\tau_2}(f(d_2))$. Let $e\in E_{\tau_1}$ and $d_1,d_2\in  \mathfrak{p}_{\tau_1}^{-1}(e)$. By induction hypothesis for Item \ref{item:glb_of_exact}, there exists $d_3\in \mathfrak{p}_{\tau_1}^{-1}(e)$ such that $d_3\leq_{\mathit{App}(E_{\tau_1})} d_1$, $d_3\leq_{\mathit{App}(E_{\tau_1})} d_2$. Since $f\in \subscript{\mathcal{E}}{\famtypes}=\mathcal{E}_{\tau_2}^{\mathcal{E}_{\tau_1}}$ is a morphism of cpo's, it is monotone. Hence, it holds that $f(d_3)\leq_{\mathit{App}(E_{\tau_2})} f(d_1)$, $f(d_3)\leq_{\mathit{App}(E_{\tau_2})} f(d_2)$. By induction hypothesis for Item \ref{item:comparable_exact}, it follows that $\mathfrak{p}_{\tau_2}(f(d_1))=\mathfrak{p}_{\tau_2}(f(d_2))$. Thus, $\mathfrak{p}_\famtypes$ is well-defined. 
		
		We now show that $\mathfrak{p}_\famtypes$ is surjective. Let $g\in \subscript{E}{\famtypes}=E_{\tau_2}^{E_{\tau_1}}$. By the induction hypothesis on  $\tau_2$ for Item \ref{item:glb_of_exact}, for each $e\in E_{\tau_1}$, we can define an element $d_e:=\glb\mathfrak{p}_{\tau_2}^{-1}(g(e))\in \mathfrak{p}_{\tau_2}^{-1}(g(e))$. By Proposition \ref{prop:consistent_aft}, for each  $a\in App(E_{\famtypes_1})\setminus\mathcal{E}_{\tau_1}$ such that there exists (at least one) $b\in \mathcal{E}_{\tau_1}$ with $b\leq_{\mathit{App}(E_{\tau_1})} a $, we can define
		\begin{equation*}
			c_a:=\lub\{d_{\mathfrak{p}_{\tau_1}(b)} \mid b\in \mathcal{E}_{\tau_1} \text{ such that }b\leq_{\mathit{App}(E_{\tau_1})} a   \}.
		\end{equation*}
		We define $f\colon  \mathit{App}(E_{\tau_1}) \to \mathit{App}(E_{\tau_2})$ for all $a\in \mathit{App}(E_{\tau_1})$ as follows:
		\begin{equation}\label{eq:preimage_construction}
			f(a):=
			\begin{cases}
				d_{\mathfrak{p}_{\tau_1}(a)} & \text{ if } a\in\mathcal{E}_{\tau_1},\\
				c_a & \text{ if } a\notin\mathcal{E}_{\tau_1} \text{ and exists }b\in\mathcal{E}_{\tau_1} \text{ such that } b\leq_{\mathit{App}(E_{\tau_1})} a,\\
				\bot_{\mathit{App}(E_{\tau_2})} & \text{ otherwise.} 
			\end{cases}
		\end{equation}
		In the following we show that $f$ is monotone. Let $a_1, a_2\in \mathit{App}(E_{\tau_1})$ such that $a_1\leq_{\mathit{App}(E_{\tau_1})} a_2$. If $a_1 \notin\mathcal{E}_{\tau_1}$ and for all $ b\in\mathcal{E}_{\tau_1}$ is not the case that  $b\leq_{\mathit{App}(E_{\tau_1})} a $, then clearly $f(a_1)=	\bot_{\mathit{App}(E_{\tau_2})} \leq_{\mathit{App}(E_{\tau_2})} f(a_2)$. If $a_1, a_2\in \mathcal{E}_{\tau_1}$, then $f(a_1)=f(a_2)$ by the induction hypothesis for Item \ref{item:comparable_exact}. If $a_1$ is exact but $a_2$ is not, then clearly $f(a_1) \leq_{\mathit{App}(E_{\tau_2})} f(a_2)$. If both $a_1, a_2\notin\mathcal{E}_{\tau_1}$ and they are greater than some exact $b\in\mathcal{E}_{\tau_1}$, then $\{ b\in \mathcal{E}_{\tau_1} \mid b\leq_{\mathit{App}(E_{\tau_1})} a_1   \}\subseteq \{ b\in \mathcal{E}_{\tau_1} \mid b\leq_{\mathit{App}(E_{\tau_1})} a_2   \}$. Hence, $f(a_1)=c_{a_1} \leq_{\mathit{App}(E_{\tau_2})} c_{a_2}= f(a_2)$, as desired. 
		It follows that $f\in\mathit{App}(\subscript{E}{\famtypes})$. Moreover, it is clear that $f$ sends exact elements to exact elements, \ie $f\in\subscript{\mathcal{E}}{\famtypes}$. For all $e\in E_{\tau_1}$, $\mathfrak{p}_\famtypes(f)(e)=\mathfrak{p}_{\tau_2}(f(c))=\mathfrak{p}_{\tau_2}(d_{\mathfrak{p}_{\tau_1}(c)})=\mathfrak{p}_{\tau_2}(d_e)=g(e)$, where $c$ is some element in the preimage $ \mathfrak{p}_{\tau_1}^{-1}(e)$. Thus, $\mathfrak{p}_\famtypes(f)=g$, as desired. 
		
		We proceed to show that Item \ref{item:comparable_exact} holds for $\famtypes$. Let $f_1, f_2\in \subscript{\mathcal{E}}{\famtypes}$ such that $f_1\leq_{\mathit{App}(E_\tau)} f_2$, and let $e\in E_{\tau_1}$. We have already shown in Item \ref{item:well-def} for $\famtypes$ that $\mathfrak{p}_\famtypes$ is well defined. In particular, $\mathfrak{p}_\famtypes(g)(e)=\mathfrak{p}_{\tau_2}(g(d_1))=\mathfrak{p}_{\tau_2}(g(d_2))$ for all $g\in\subscript{\mathcal{E}}{\famtypes}$ and $d_1, d_2\in \mathfrak{p}_{\tau_1}^{-1}(e)$. By the definition of the order on morphisms, $f_1(d)\leq_{\mathit{App}(E_{\tau_2})} f_2(d)$ for all $d\in\mathcal{E}_{\tau_1}$. By the induction hypothesis for Item \ref{item:comparable_exact}, it holds that $\mathfrak{p}_\famtypes(f_1)(e)=\mathfrak{p}_{\tau_2}(f_1(d))= \mathfrak{p}_{\tau_2}(f_2(d))=\mathfrak{p}_\famtypes(f_2)(e)$. Hence, $\mathfrak{p}_\famtypes(f_1)=\mathfrak{p}_\famtypes(f_2)$, as desired. 
		
		Finally, we show that Item \ref{item:glb_of_exact} holds for $\famtypes$. Let $g\in \subscript{E}{\famtypes}$.
		We can construct a morphism $f\in  \subscript{\mathcal{E}}{\famtypes}$ using the same technique as in \eqref{eq:preimage_construction}. By the proof of Item \ref{item:surjective}, we already have $\mathfrak{p}_\famtypes(f)=g$. It remains to show that $f=\glb\mathfrak{p}_\famtypes^{-1}(g)$. Let $h\in \mathfrak{p}_\famtypes^{-1}(g)$. First notice that since $\mathfrak{p}_\famtypes(h)=\mathfrak{p}_\famtypes(f)=g$, it holds that for all $e\in E_{\tau_1}$, $\mathfrak{p}_{\tau_2}(f(l))=\mathfrak{p}_{\tau_2}(h(l))=g(e)$ for all $l\in \mathfrak{p}_{\tau_1}^{-1}(e)$. In particular, for all $a\in \mathcal{E}_{\tau_1}$, it holds that $f(a), h(a)\in \mathfrak{p}_{\tau_2}^{-1}(g(\mathfrak{p}_{\tau_1}(a)))$.  Hence, for all $a\in \mathcal{E}_{\tau_1}$, we have that $f(a)=d_{\mathfrak{p}_{\tau_1}(a)}=\glb(\mathfrak{p}_{\tau_2}^{-1}(g(\mathfrak{p}_{\tau_1}(a))))\leq_{\mathit{App}(E_{\tau_2})} h(a)$. Now let $a\notin\mathcal{E}_{\tau_1}$  such that there exists $b\in\mathcal{E}_{\tau_1}$ such that $b\leq_{\mathit{App}(E_{\tau_1})} a$. Since we have already shown that $f(c)=d_{\mathfrak{p}_{\tau_1}(c)}\leq_{\mathit{App}(E_{\tau_2})} h(c)$ for all $c\in \mathcal{E}_{\tau_1}$, it is easy to see that $f(a)=c_a\leq_{\mathit{App}(E_{\tau_2})} h(a)$. For all the other cases of $a\in \mathit{App}(E_{\tau_1})$ it is obvious that $f(a)\leq_{\mathit{App}(E_{\tau_2})} h(a)$. Hence, $f$ is a lower bound of $\mathfrak{p}_\famtypes^{-1}(g)$. Since $f\in \mathfrak{p}_\famtypes^{-1}(g)$, we get $f=\glb\mathfrak{p}_\famtypes^{-1}(g)$, as desired.
	\end{proof}

		In most applications of AFT, for approximation spaces of base types, there exists a unique exact element representing an object of a semantics, and Items \ref{item:comparable_exact} and \ref{item:glb_of_exact} of Proposition \ref{prop:proj_well-def} are trivially verified. However, for higher-order approximation spaces, this is not always the case, as we illustrate in the following example. 
		
		\begin{example}\label{example:two_exact}
			Let $o$ be the Boolean type, with semantics $E_o:=\langle \{\lfalse, \ltrue\}, \leq_t\rangle$, where $\leq_t$ is the standard truth order. In standard AFT, we would define the approximation space for $E_o$ to be the bilattice $\mathit{App}(E_o):=\langle E_o \times E_o, \leqp\rangle$, with $\leqp$ the precision order. Then, the semantics for $o\to o$ is the poset of functions from $E_o$ to $E_o$, and the approximation space for it is the exponential, \ie the set of monotone functions from $\mathit{App}(E_o)$ to itself, ordered pointwise. Clearly, we can set the exact elements of $\mathit{App}(E_o)$ to be $(\lfalse, \lfalse)$ and $(\ltrue, \ltrue)$, and $\mathfrak{p}_o$ to send them to $\lfalse$ and $\ltrue$, respectively. Now consider the following two functions: $f, g\colon \mathit{App}(E_o)\to \mathit{App}(E_o)$ defined by $f(\lfalse,\ltrue)=(\lfalse,\ltrue)$, $g(\lfalse,\ltrue)=f(\lfalse,\lfalse)=g(\lfalse,\lfalse)=f(\ltrue,\ltrue)=g(\ltrue,\ltrue)=(\ltrue,\ltrue)$, and $f(\ltrue,\lfalse)=g(\ltrue,\lfalse)=(\ltrue,\lfalse)$. 
			Clearly, both $f$ and $g$ send exacts to exacts, thus, they are exact. Moreover, even though $f\neq g$, it is easy to see that $\mathfrak{p}_{o\to o}(f)=\mathfrak{p}_{o\to o}(g)=h\colon E_o\to E_o$, where $h(\lfalse)=h(\ltrue)=\ltrue$.
		\end{example}

	We conclude this section with the definition of \emph{consistent} elements.
		\begin{definition}
			Let $\subscript{E}{\famtypes}\in S_\mathbb{T}$. An element $c\in \mathit{App}(\subscript{E}{\famtypes})$ is \emph{consistent} if there exists $e\in \subscript{\mathcal{E}}{\famtypes}$ such that $c\leq_{ \mathit{App}(\subscript{E}{\famtypes})} e$.
		\end{definition}
	
		
		Notice that a function of the family $\{\mathfrak{p}_\famtypes\colon \subscript{\mathcal{E}}{\famtypes} \to \subscript{E}{\famtypes} \}_{\subscript{E}{\famtypes}\in S_\mathbb{T}}$ not only determines which element of the semantics an exact element represents, but it also helps understanding what a consistent element is approximating: if $c\in\mathit{App}(\subscript{E}{\famtypes})$ is consistent and $c\leq_{\mathit{App}(\subscript{E}{\famtypes})} e$ for some exact $e$, then $c$ approximates $\mathfrak{p}_\famtypes(e)$. Clearly, consistent elements may approximate more than one element of a semantics.
		
%
%

		
		\section{An approximation system for standard AFT}\label{sec:standardAFT}

		In this section, we show how our new framework extends the standard AFT setting to higher-order definitions.

		The main building block of an approximation system is the category $\Approx$, containing all the desired approximation spaces. Hence, we start by showing that the approximation spaces used in standard AFT, i.e.\ the square bilattices, form a Cartesian closed category.
		
		
		First, recall that a \emph{square bilattice} is a poset of the form $\langle L\times L, \leqp \rangle$, where $\langle L, \leq_L\rangle$ is a complete lattice and $\leqp$ is the \emph{precision order}, i.e.\ $(x_1, y_1)\leqp (x_2, y_2)$ iff $x_1\leq_L x_2$ and $y_2\leq_L y_1$. If we view these objects from a category-theoretic perspective, we can write $\langle L\times L, \leqp \rangle =\mathcal{L}\times \op{\mathcal{L}}$ where $\mathcal{L}:=\langle L, \leq_L\rangle \in \Obj{\CLat}$. Hence, we can define the category $\BiLat$ of square bilattices 
		as follows:
		
		\begin{align*}
			\Obj{\BiLat}:=&\left\{\mathcal{L}\times\op{\mathcal{L}}\mid\mathcal{L}\in\CLat\right\}
			\\ \Hom{\BiLat}:=&\left\{f\colon\mathcal{L}_1\to\mathcal{L}_2\mid \mathcal{L}_1, \mathcal{L}_2\in\Obj{\BiLat} \wedge f \text{ monotone}\right\}
		\end{align*}
		
		We will denote an element $\mathcal{L}\times\op{\mathcal{L}}$ of $\BiLat$ by $\overline{\mathcal{L}}$. 
		
		\begin{lemma}\label{lemma:BiLat_fullsubcat}
			The category $\BiLat$ is a full subcategory of $\CLat$.
		\end{lemma}
		\begin{proof}
			Clearly, if $\mathcal{L}\in\CLat$, then $\op{\mathcal{L}}\in\CLat$. Since $\CLat$ is Cartesian closed, for all $\mathcal{L}\in \CLat$, we have that $\mathcal{L}\times \op{\mathcal{L}}\in \CLat$. We conclude by the definition of $\BiLat$.
		\end{proof}

		By Lemma \ref{lemma:BiLat_fullsubcat}, proving that $\BiLat$ is Cartesian closed reduces to show that the following isomorphisms of complete lattices hold for all $\overline{\mathcal{L}_1},\overline{\mathcal{L}_2}\in \BiLat$: 
		\begin{enumerate}
			\item\label{item:iso_terminal_obj} $\mathcal{T}\cong \overline{\mathcal{T}}$, where $\mathcal{T}$ is the terminal object of $\CLat$,
			\item\label{item:iso_product} $\overline{\mathcal{L}_1}\times \overline{\mathcal{L}_2}\cong \overline{(\mathcal{L}_1\times\mathcal{L}_2)}$,
			\item\label{item:iso_exponential} $\overline{\mathcal{L}_2}^{\overline{\mathcal{L}_1}}\cong \overline{\mathcal{L}_2^{\overline{\mathcal{L}_1}}}$. 
		\end{enumerate} 
		
		While the first two isomorphisms are rather straightforward, the latter deserves some attention. Consider a morphism of square bilattices $f$ from $\overline{\mathcal{L}_1}$ to $\overline{\mathcal{L}_2}$. Since $\overline{\mathcal{L}_2}=\mathcal{L}_2\times \op{\mathcal{L}_2}$, we can write $f$ as a pair $(f_1,f_2)$ of morphisms of complete lattices, where $f_1\colon \overline{\mathcal{L}_1}\to \mathcal{L}_2$ and $f_2\colon \overline{\mathcal{L}_1}\to \op{\mathcal{L}_2}$. It follows easily that $\overline{\mathcal{L}_2}^{\overline{\mathcal{L}_1}}\cong \mathcal{L}_2^{\overline{\mathcal{L}_1}}\times (\op{\mathcal{L}_2})^{\overline{\mathcal{L}_1}}$. Then, the isomoprhism $\varphi\colon \overline{\mathcal{L}_2}^{\overline{\mathcal{L}_1}}\to \overline{\mathcal{L}_2^{\overline{\mathcal{L}_1}}}$, is realised by mapping  $f\hat{=}(f_1,f_2)$ to a new pair $\varphi(f):=(f_1,f_2')\in \mathcal{L}_2^{\overline{\mathcal{L}_1}}\times \op{(\mathcal{L}_2^{\overline{\mathcal{L}_1}})}=\overline{\mathcal{L}_2^{\overline{\mathcal{L}_1}}}$, where the second component is defined by $f_2'(x,y):=f_2(y,x)$. Notice that, since $f_2$ is a monotone function from $\overline{\mathcal{L}_1}$ to $\op{\mathcal{L}_2}$, $f_2'$  is indeed a monotone function from $\overline{\mathcal{L}_1}$ to $\mathcal{L}_2$. Further details regarding the isomorphisms listed above are contained in the proof of Theorem \ref{thm:Approx_cartclosed}.

		\begin{theorem}\label{thm:Approx_cartclosed}
			The category $\BiLat$ is Cartesian closed.
		\end{theorem}
		\begin{proofwithoutqed}
			Since $\CLat$ is Cartesian closed, and $\BiLat$ is a full-subcategory of $\CLat$ (Lemma \ref{lemma:BiLat_fullsubcat}), it is sufficient to show that the terminal object of $\CLat$ is an object of $\BiLat$ and that for all $\overline{\mathcal{L}_1}, \overline{\mathcal{L}_2}\in\BiLat$, the product $\overline{\mathcal{L}_1}\times\overline{\mathcal{L}_2}$ and the exponential $\overline{\mathcal{L}_2}^{\overline{\mathcal{L}_1}}$, computed in the category $\CLat$, are also objects of $\BiLat$.
			\begin{itemize}
				\item \emph{Terminal object.} There is an obvious isomorphism from the terminal object $\mathcal{T}$ of $\CLat$, i.e.\ the lattice with just one element and trivial order, and the object $\mathcal{T}\times\op{\mathcal{T}}\in\BiLat$.
				\item \emph{Product.} Let $\overline{\mathcal{L}_1}, \overline{\mathcal{L}_2}\in\BiLat$. By Cartesian closedness of $\CLat$, $\overline{\mathcal{L}_1}\times\overline{\mathcal{L}_2}$ is an object of $\BiLat$. We define a function $\varphi\colon \overline{\mathcal{L}_1}\times\overline{\mathcal{L}_2} \to \overline{\mathcal{L}_1\times\mathcal{L}_2}$ by sending an element $((a_1,b_1),(a_2,b_2))$ to $((a_1,a_2),(b_1,b_2))$. Clearly, $\varphi$ is bijective. Moreover, by the definition of the product order, the following double-implications hold for all $a_1,b_1,x_1,y_1\in \mathcal{L}_1$, and for all $b_1,b_2, x_2,y_2\in \mathcal{L}_2$
				\begin{equation*}
					\begin{split}
						((a_1,b_1),(a_2,b_2))&\leq_{\overline{\mathcal{L}_1}\times\overline{\mathcal{L}_2}}  ((x_1,y_1),(x_2,y_2))\\
						&\iff (a_1,b_1)\leq_{\overline{\mathcal{L}_1}} (x_1,y_1) \wedge (a_2,b_2)\leq_{\overline{\mathcal{L}_2}} (x_2,y_2) \\
						& \iff a_1\leq_{\mathcal{L}_1} x_1 \wedge y_1\leq_{\mathcal{L}_1} b_1 \wedge a_2\leq_{\mathcal{L}_2} x_2 \wedge y_2\leq_{\mathcal{L}_2} b_2 \\
						& \iff (a_1,a_2)\leq_{\mathcal{L}_1\times \mathcal{L}_2} (x_1,x_2) \wedge (y_1,y_2)\leq_{\mathcal{L}_1\times\mathcal{L}_2} (b_1,b_2) \\
						&\iff 	((a_1,a_2),(b_1,b_2))\leq_{\overline{\mathcal{L}_1\times\mathcal{L}_2} } ((x_1,x_2),(y_1,y_2)).
					\end{split}
				\end{equation*}
				Hence, $\varphi$ and its inverse are monotone functions, i.e.\ morphisms. It follows that $\overline{\mathcal{L}_1}\times\overline{\mathcal{L}_2} \cong \overline{\mathcal{L}_1\times\mathcal{L}_2}\in\BiLat$, as desired.
				\item \emph{Exponential.} Let $\overline{\mathcal{L}_1}, \overline{\mathcal{L}_2}\in\BiLat$. By Cartesian closedness of $\CLat$, $\overline{\mathcal{L}_2}^{\overline{\mathcal{L}_1}}$ is an object of $\BiLat$. Let $\delta\colon \overline{\mathcal{L}_1}\to\overline{\mathcal{L}_1}$ be the function sending $(x,y)$ to $(y,x)$. We define a function $\psi\colon \overline{\mathcal{L}_2}^{\overline{\mathcal{L}_1}} \to \overline{\mathcal{L}_1^{\overline{\mathcal{L}_2}}}$ by sending a morphism $f:=(f_1,f_2)$ to $(f_1,f_2\circ\delta)$, where $f_1\colon \overline{\mathcal{L}_1}\to \mathcal{L}_2$ and $f_2\colon \overline{\mathcal{L}_1}\to \op{\mathcal{L}_2}$ are the components of $f$. Since $f_2$ is an antimonotone function from $\overline{\mathcal{L}_1}$ to $\mathcal{L}_2$, it is easy to check that $f_2\circ\delta$ is a monotone function from $\overline{\mathcal{L}_1}$ to $\mathcal{L}_2$, as desired. Clearly, $\varphi$ is bijective.  Moreover, by the definition of the pointwise order, the following double-implications hold for all $f_1,g_1\in \mathcal{L}_2^{\overline{\mathcal{L}_1}}$, and for all $f_2,g_2\in (\op{\mathcal{L}_2})^{\overline{\mathcal{L}_1}}$
				\begin{equation*}
					\begin{split}
						(f_1,f_2) &\leq_{\overline{\mathcal{L}_2}^{\overline{\mathcal{L}_1}}}  (g_1,g_2)\\
						&\iff \forall (x,y)\in\overline{\mathcal{L}_1},  (f_1(x,y),f_2(x,y))\leq_{\overline{\mathcal{L}_2}} (g_1(x,y),g_2(x,y))\\
						& \iff \forall (x,y)\in\overline{\mathcal{L}_1},  f_1(x,y)\leq_{\mathcal{L}_2} g_1(x,y) \wedge g_2(x,y)\leq_{\mathcal{L}_2} f_2(x,y)\\
						& \iff \forall (x,y)\in\overline{\mathcal{L}_1},  f_1(x,y)\leq_{\mathcal{L}_2} g_1(x,y) \wedge g_2(y,x)\leq_{\mathcal{L}_2} f_2(y,x)\\
						& \iff f_1\leq_{\mathcal{L}_2^{\overline{\mathcal{L}_1}}} g_1 \wedge g_2\circ\delta\leq_{\mathcal{L}_2^{\overline{\mathcal{L}_1}}} f_2\circ\delta \\
						&\iff 	(f_1, f_2\circ \delta) \leq_{\overline{\mathcal{L}_1^{\overline{\mathcal{L}_2}}}} (g_1, g_2\circ\delta).
					\end{split}
				\end{equation*}
				Hence, $\psi$ and its inverse are monotone functions, i.e.\ morphisms. It follows that $\overline{\mathcal{L}_2}^{\overline{\mathcal{L}_1}} \cong \overline{\mathcal{L}_1^{\overline{\mathcal{L}_2}}}\in\BiLat$, as desired.\qed
			\end{itemize}
		\end{proofwithoutqed}
		
		It is interesting to observe that the approximators used in standard AFT, \ie the \emph{symmetric} approximators, when viewed in their square bilattice approximator space, correspond to pairs of equal functions, \ie the classic definition of exact pair \citep{DMT00ApproximationsStableOperatorsWell-FoundedFixpointsApplications}. Similarly, a \emph{gracefully degrading} approximator $A=(A_1, A_2)\colon \overline{\mathcal{L}}\to\overline{\mathcal{L}}$, \ie such that $A_1(x,y)\leq_{\overline{\mathcal{L}}} A_2(y,x)$ for all $(x,y)\in\overline{\mathcal{L}}$ \citep{DV07Well-FoundedSemanticsAlgebraicTheoryNon-monotoneInductive}, when viewed in $\overline{\mathcal{L}^{\overline{\mathcal{L}}}}$ is a pair $\varphi(A)=(A_1,A'_2)$ with $A_1\leq_{\mathcal{L}^{\overline{\mathcal{L}}}} A'_2$, \ie a consistent pair according to the classic definition of AFT.

		Thanks to Theorem \ref{thm:Approx_cartclosed} and $\BiLat\subseteq\CPO$, we can fix $\BiLat$ as our approximation category. 
		This can be done for \emph{any} application in which we want to use standard AFT techniques. Nevertheless, depending on the application at hand, the approximation system may differ. Let us show how to define an approximation system given a language 
		based on a type hierarchy $\mathbb{H}$
		. Let $S_\mathbb{T}$ be the set of the semantics of types of $\mathbb{H}$ we want to approximate, and assume that such semantics are complete lattices, as is usually the case in logic programming. Then, we can inductively define a mapping $\mathit{App}\colon S_\mathbb{T} \to \Obj{\BiLat}$ by setting, for all $\tau\in\mathcal{B}_\mathbb{T}$, $\mathit{App}(E_\tau):=\overline{E_\tau}$, and proceed using the conditions in Definition \ref{def:approximation_system}.
		Notice that the base case of the induction is nothing more than what is usually done in standard AFT: from a complete lattice $\langle L, \leq_L\rangle$ we obtain the square bilattice $\langle L^2, \leqp\rangle$. The remaining steps are naturally provided by following the Cartesian closed structure of $\BiLat$.
		
		For each base type $\tau\in\mathcal{B}_\mathbb{T}$, the exact elements of $\mathit{App}(E_\tau)$ are defined as in standard AFT: $(x,y)\in \mathit{App}(E_\tau)$ is exact if $x=y$, \ie $\mathcal{E}_\tau=\{(x,x)\mid x\in E_\tau\}$. Notice that, since $\BiLat\subseteq\CLat\subseteq \compjoinsl$, the condition \ref{item:condition_completejoinsemilattice} in Defintion \ref{def:approximation_system} is satisfied. 
		Finally, 
		for each base type $\tau$, we define $\zeroproj{\tau}\colon \mathcal{E}_\tau\to E_\tau$ by sending $(x,x)$ to $x$. Both conditions \ref{item:condition_comparable_exact} and \ref{item:condition_glb_exact} in Definition \ref{def:approximation_system} hold since $(\zeroproj{\tau})^{-1}(x)=\{(x,x)\}$.
		Hence, we have obtained an approximation system $(\BiLat, \mathit{App}, \{\mathcal{E}_\tau\}_{\tau\in\mathcal{B}}, \{\zeroproj{\tau}\}_{\tau\in\mathcal{B}})$ for $S_\mathbb{T}$. 
		
		In standard AFT, we are ultimately interested in the approximation space of interpretations. Given a vocabulary $V$, $S_\mathbb{T}$ can be easily chosen to contain the semantics of the types of the symbols in $V$ and the space of interpretations for $V$, \ie the complete lattice $\Pi_{s\in V'} E_{t(s)}$, where $t(s)$ is the type of the symbol $s$.
		It follows that the approximation space of interpretations is  $\mathit{App}(\Pi_{s\in V'} E_{t(s)})=\Pi_{s\in V'} \mathit{App}(E_{t(s)})\in\BiLat$.
		Clearly, if we restrict to a vocabulary with only symbols of base type, then we retrieve the usual framework of 
		standard AFT.

			\section{Revised Extended Consistent AFT}\label{sec:iclp18}
		

			\citet{CRS18ApproximationFixpointTheoryWell-FoundedSemanticsHigher-Order} developed an extension of consistent AFT \citep{DMT03Uniformsemantictreatmentdefaultautoepistemiclogics} to generalize the well-founded semantics for classical logic programs to one for programs with higher-order predicates. As already pointed out in Section \ref{sec:AFT}, this generalization bears some issues.
			
			In this section, we examine in detail the work of \citet{CRS18ApproximationFixpointTheoryWell-FoundedSemanticsHigher-Order} under the lenses of our novel categorical framework. First, in Subsection \ref{sec:sub:approxcat_greeks}, we present their extension of consistent AFT with their version of approximation spaces, and we prove that this new class of mathematical objects forms a Cartesian closed category. Then, in Subsection \ref{sec:sub:approxsystem_greeks}, we briefly recall the types and semantics used by  \citet{CRS18ApproximationFixpointTheoryWell-FoundedSemanticsHigher-Order}, and we define an approximation system for it. Thanks to the inductive nature of Cartesian closed categories, from the tuple defining the approximation system, we can effortlessly retrieve the entire, complex hierarchy built by \citet{CRS18ApproximationFixpointTheoryWell-FoundedSemanticsHigher-Order}. From the definition of the approximation system, we already obtain a concept of exactness for higher-order objects, which was previously missing in \cite{CRS18ApproximationFixpointTheoryWell-FoundedSemanticsHigher-Order}.  Finally, in Subsection \ref{sec:sub:approximator_greeks}, we present our solution to the problem encountered in the work of \citet{CRS18ApproximationFixpointTheoryWell-FoundedSemanticsHigher-Order} concerning the semantics of logic programs with existential quantifiers. In particular, we propose a new approximator which provides the expected well-founded semantics. We conclude the subsection with two examples of  logic programs in which we need to apply an approximate object on another approximate object.

			\subsection{The Approximation Category for  Extended Consistent AFT}\label{sec:sub:approxcat_greeks}
				In consistent AFT \citep{DMT03Uniformsemantictreatmentdefaultautoepistemiclogics}, an approximation space is  the consistent part of a square bilattice, \ie given a bilattice $\overline{\mathcal{L}}=\langle L\times L, \leqp \rangle $, only the subset $\{(x,y)\mid x\leq_L y\}\subseteq \overline{\mathcal{L}}$ of consistent elements is taken into account.  \cite{CRS18ApproximationFixpointTheoryWell-FoundedSemanticsHigher-Order} extended consistent AFT to a new class of approximation spaces: the sets of the form $L\otimes U:=\{(x,y)\mid x\in L, y\in U, x\leq y\}$, comprising the consistent elements of the cartesian product between a set $L$ of \emph{lower bounds} and a set $U$ of \emph{upper bounds}, where $L$ may differ from $U$. 
			 	
			 	In order for the machinery of consistent AFT to work over these new spaces, \cite{CRS18ApproximationFixpointTheoryWell-FoundedSemanticsHigher-Order} added  some conditions to restrain the possible choices for $L$ and $U$.

			 \begin{definition}\label{def:approximation_tuple}
			 	An \emph{approximation tuple} is a tuple $(L,U\leq)$, where $L$ and $U$ are sets, and $\leq$ is a partial order on $L\cup U$ such that the following conditions  hold:
			 	\begin{enumerate}
			 		\item $\langle L\cup U, \leq \rangle$ has a top element $\top$ and a bottom element $\bot$,
			 		\item $\top, \bot\in L\cap U$,
			 		\item $\langle L, \leq \rangle$ and $\langle U, \leq \rangle$ are complete lattices,
			 		\item \emph{Interlattice Least Upper Bound Property (ILP)}: for all $b\in U$ and for all $S\subseteq L$ such that for all $x\in S$, $x\leq b$, we have $\lub_L S\leq b$,
			 		\item \emph{Interlattice Greatest Lower Bound Property (IGP)}: for all $a\in L$ and for all $S\subseteq U$ such that for all $x\in S, a\leq x$, we have $a\leq \glb_U S$.
			 	\end{enumerate}
			 \end{definition}
			 
			 \begin{definition}\label{def:approximation_space_LU}
			 	Let $(L,U, \leq )$ be an approximation tuple. The \emph{approximation space} (associated to $(L,U,\leq)$) is the poset $\langle L\otimes U, \leqp\rangle$, where
			 	$L\otimes U:=\{(x,y)\mid x\in L, y\in U, x\leq y\}$,
			 	and $\leqp$ is the partial order defined for all $(x_1,y_1),(x_2,y_2)\in L\otimes U$ by:
			 	$(x_1,y_1)\leqp (x_2,y_2)$ iff $x_1\leq x_2$ and $ y_2\leq y_1$.
			 	We call $\leqp$ the \emph{precision order} on $L\otimes U$.
			 \end{definition}
			 
			 In the remainder of this subsection, we prove that the new class of approximation spaces defined in Definition \ref{def:approximation_space_LU} forms a Cartesian closed full subcategory of $\CPO$ (Theorem \ref{thm:LU_approx_category}).  First, we define a new category $\greeks$, with objects the approximation spaces just introduced, as follows:
			 
			 \begin{equation*}
			 	\begin{split}
			 		\Obj{\greeks}:=& \{ \langle L\otimes U, \leqp\rangle\mid (L,U,\leq) \text{ is an approximation tuple} \}\\
			 		\Hom{\greeks}:=& \{ f\colon A\to B\mid A,B\in\Obj{\greeks} \wedge f \text{ monotone}\}.
			 	\end{split}
			 \end{equation*}


			\begin{theorem}\label{thm:LU_approx_category}
				The category $\greeks$ is a Cartesian closed full subcategory of $\CPO$.
			\end{theorem}
			
				We split the proof of Theorem \ref{thm:LU_approx_category} into smaller results: first we show that $\greeks$ is a full subcategory of $\CPO$, then we prove it is Cartesian closed.
			
			\begin{proposition}\label{prop:approx_is_cpo}
				Let $L\otimes U\in \Obj{\greeks}$. Then $L\otimes U$ is a cpo. 
			\end{proposition}
			\begin{proof}
				Let $L\otimes U\in \Obj{\greeks}$, and $S\subseteq L\otimes U$ a chain. We denote by $p_1\colon L\otimes U\to L$ the function of sets sending $(x,y)$ to $x$, and by $p_2\colon L\otimes U\to U$ the function sending $(x,y)$ to $y$.  Clearly, $p_1(S)$ and $p_2(S)$ are chains in $\langle L, \leq\rangle$ and $\langle U, \leq \rangle$, respectively. Since $\langle L, \leq\rangle$ and $\langle U, \leq \rangle$ are lattices, there exist $\lub_L p_1(S)=:x\in p_1(S)$ and $\glb_U p_2(S)=:y\in p_2(S)$. We now show that $(x,y)\in L\otimes U$, \ie $x\leq y$. Let $r\in p_1(S)\subseteq L$ and $q\in p_2(S)\subseteq U$. Then, there exist $p\in p_1(S)$ and $s\in p_2(S)$ such that $(r,s),(p,q)\in S$. Since $S$ is a chain, we either have $(r,s)\leqp(p,q)$ or $(p,q)\leqp(r,s)$. In both cases, $s\leq q$. By the arbitrarity of $q$ and the IGP, $s\leq y$. By the arbitrarity of $s$ and the ILP, we have $x\leq y$, as desired. Clearly $(x,y)=\lub_{L\otimes U} S\in L\otimes U$, so it remains to show that $(x,y)\in S$. Since $x\in p_1(S)$ and $y\in p_2(S)$, there exist $x'\in p_1(S)$ and $y'\in p_2(S)$ such that $(x,y'), (x',y)\in S$. By the definitions of $x$ and $y$, we must have $x'\leq x$ and $y'\geq y$. Suppose $x' < x$ and $y' > y$. Then $(x,y')$ and $(x,y')$ cannot be ordered, which condradicts $S$ being a chain. Hence, either $x=x'$ or $y=y'$. In any case, $(x,y)\in S$, as desired. 
			\end{proof}
			
			\begin{corollary}\label{cor:Approx_fullsubcat_cpo}
				$\greeks$ is a full subcategory of $\CPO$. 
			\end{corollary}
			\begin{proof}
				Clear from Proposition \ref{prop:approx_is_cpo} and the definition of $\Hom{\greeks}$.
			\end{proof}

			\begin{proposition}
				$\greeks$ is a Cartesian closed category.
			\end{proposition}
			\begin{proofwithoutqed}
				Since $\CPO$ is Cartesian closed, and $\greeks$ is a full-subcategory of $\CPO$ (Corollary \ref{cor:Approx_fullsubcat_cpo}), it is sufficient to show that the terminal object of $\CPO$ is an object of $\greeks$ and that for all $\mathcal{A}, \mathcal{B}\in\greeks$, the product $ \mathcal{A}\times\mathcal{B}$ and the exponential $\mathcal{B}^\mathcal{A}$, computed in the category $\CPO$, are also objects of $\greeks$.
				
				Let $\mathcal{A}, \mathcal{B}\in\greeks$.
				Let $(L_\mathcal{A}, U_\mathcal{A}, \leq_\mathcal{A})$ and $(L_\mathcal{B}, U_\mathcal{B}, \leq_\mathcal{B})$ be the approximation tuples of the approximation spaces $\mathcal{A}, \mathcal{B}\in \greeks$. 
				We denote the orders on $\mathcal{A}$ and $\mathcal{B}$ as $\leq_{p,\mathcal{A}}$ and $\leq_{p,\mathcal{B} }$, respectively. 
				\begin{itemize}
					\item \emph{Terminal object.} Let $\mathcal{T}=\langle \{\ast\}, \leq \rangle$ be the cpo with one element, \ie the terminal object of $\CPO$. Clearly, $(\{\ast\}, \{\ast\}, \leq)$ is an approximation tuple, and $\mathcal{T}\cong\langle \{\ast\}\otimes\{\ast\}, \leq\rangle$ in $\CPO$. Hence, $\mathcal{T}\in\Obj{\greeks}$.
					\item \emph{Product.} This follows from Proposition \ref{prop:full_subcat_poset_genproducts} and Corollary \ref{cor:Approx_fullsubcat_cpo}.  

					\item \emph{Exponential.} We have to show that $\mathcal{B}^\mathcal{A}\in\CPO$ is isomorphic (in $\CPO$) to some $\mathcal{C}=L_\mathcal{C}\otimes U_\mathcal{C}\in \greeks$. Let $L_\mathcal{C}:=\hom_\CPO(\mathcal{A}, \langle L_\mathcal{B},\leq_\mathcal{B}\rangle)$, $U_\mathcal{C}:=\hom_\CPO(\mathcal{A}, \langle U_\mathcal{B},\geq_\mathcal{B}\rangle)$, and $\leq_\mathcal{C}$ be the restriction onto $L_\mathcal{C}\cup\ U_\mathcal{C}$ of the pointwise extension of $\leq_\mathcal{B}$, namely for all $f,g\in L_\mathcal{C}\cup U_\mathcal{C}$,
					\begin{align*}
						f\leq_\mathcal{C} g \iff \forall x\in \mathcal{A}, f(x)\leq_\mathcal{B} g(x)
					\end{align*}
					We first show that $L_\mathcal{C}\otimes U_\mathcal{C}$ with the precision order $\leq_{p,\mathcal{C}}$ induced by  $\leq_\mathcal{C}$ is an object of $\greeks$. In other words, we show that $(L_\mathcal{C}, U_\mathcal{C}, \leq_\mathcal{C})$ is an approximation tuple.
					\begin{enumerate}
						\item The morphisms $\bot_\mathcal{C}\colon x\mapsto \bot_\mathcal{B}$ and $\top_\mathcal{C}\colon x\mapsto \top_\mathcal{B}$ are the bottom and top element of $L_\mathcal{C}\cup U_\mathcal{C}$, respectively.
						\item Clearly, $\bot_\mathcal{C},\top_\mathcal{C}\in L_\mathcal{C}\cap U_\mathcal{C}$.
						\item Since $\langle L_\mathcal{B}, \leq_\mathcal{B} \rangle$ and $\langle U_\mathcal{B}, \leq_\mathcal{B} \rangle$ are complete lattices by definition of approximation space, and $\leq_\mathcal{C}$ is the pointwise extension of $\leq_\mathcal{B}$, it is straightforward to see that $\langle L_\mathcal{C}, \leq_\mathcal{C} \rangle$ and $\langle U_\mathcal{C}, \leq_\mathcal{C} \rangle$ are also complete lattices. 
						\item \label{item:ILP_approx_exp_obj} Let $g\in U_\mathcal{C}$, and let $S\subseteq L_\mathcal{C}$ such that  for all $f\in S$, $f\leq_\mathcal{C} g$, \ie for all $x\in \mathcal{A}$ we have $f(x)\leq_\mathcal{B} g(x)$. We have to show that $\lub_{L_\mathcal{C}} S\leq_\mathcal{C} g$. Since $g\in U_\mathcal{C}$ and $f \in L_\mathcal{C}$, for all $x\in\mathcal{A}$ we have $g(x) \in U_\mathcal{B}$ and $S_x:=\{f(x)\mid f\in S\}\subseteq L_\mathcal{B}$. By using the ILP on $\mathcal{B}$, we get that $\lub_{L_\mathcal{B}} S_x \leq_\mathcal{B} g(x)$, for all $x\in \mathcal{A}$. It is not difficult to see that $\lub_{L_\mathcal{C}} S(x)=\lub_{L_\mathcal{B}} S_x$, for all $x\in\mathcal{A}$. Hence, $\lub_{L_\mathcal{C}} S\leq_\mathcal{C} g$, as desired.
						\item Analogous to Item \ref{item:ILP_approx_exp_obj}.
					\end{enumerate}
					Hence, $\mathcal{C}\in\greeks$. It remains to show that $\mathcal{C}$ is isomorphic to $\mathcal{B}^\mathcal{A}$ in $\CPO$. Notice that there is an obvious isomorphism of sets
					\begin{equation*}
						\begin{split}
							\mu\colon  \hom_\CPO(\mathcal{A}, L_\mathcal{B}\times U_\mathcal{B})& \to L_\mathcal{C}\times U_\mathcal{C} \\
							f & \mapsto (f_1,f_2),
						\end{split}
					\end{equation*}
					where $f_1,f_2$ are the two components of $f$. By the definitions of the orders (notice the inversion of the order $\leq_\mathcal{B}$ on $U_\mathcal{B}$ in $U_\mathcal{C}$), it is easy to check that $\mu$ and $\mu^{-1}$ are both well-defined, \ie they send a monotone function to a pair of monotone functions, and a pair of monotone functions to a monotone function, respectively. Now, let $f\in \hom_\CPO(\mathcal{A}, L_\mathcal{B}\times U_\mathcal{B})$. Then
					\begin{equation*}
						\begin{split}
							f\in \hom_\CPO(\mathcal{A}, \mathcal{B}) & \iff \forall x\in \mathcal{A}, f(x)\in\mathcal{B}\\
							& \iff \forall x\in \mathcal{A}, \mu_1(f)(x)=f_1(x)\leq_\mathcal{B} f_2(x)=\mu_2(f)(x)\\
							& \iff \mu_1(f)\leq_\mathcal{C} \mu_2(f)\\
							& \iff \mu(f) \in L_\mathcal{C}\otimes U_\mathcal{C}.
						\end{split}
					\end{equation*}
					Analogously, if $(g,h)\in L_\mathcal{C}\otimes U_\mathcal{C}$, then $\mu^{-1}(f,g)\in \hom_\CPO(\mathcal{A}, \mathcal{B})$. Hence, by restricting domain and codomain of $\mu$, we get another isomorphism of sets $\nu\colon \hom_\CPO(\mathcal{A}, \mathcal{B}) \to L_\mathcal{C}\otimes U_\mathcal{C}$. It remains to show that $\nu$ and $\nu^{-1}$ are monotone. Let $f,g \in \hom_\CPO(\mathcal{A}, \mathcal{B}) $ such that $f\leq_{\mathcal{B}^\mathcal{A}} g$, \ie for all $x\in \mathcal{A}$, $f(x)\leq_{p,\mathcal{B}} g(x)$. By the defintion of the precision order, this means that $\nu_1(f)(x) \leq_\mathcal{B} \nu_1(g)(x) \leq_\mathcal{B} \nu_2(g)(x) \leq_\mathcal{B} \nu_2(f)(x)$ for all $x\in\mathcal{A}$. Hence, $\nu_1(f)\leq_\mathcal{C} \nu_1(g)$ and $\nu_2(g)\leq_\mathcal{C} \nu_2(f)$. It follows by definition that $\nu(f)\leq_{p,\mathcal{C}} \nu(g)$, as desired. The analogous result holds for $\nu^{-1}$ and can be shown similarly. Therefore, the corresponding morphism $\nu'\colon \mathcal{B}^\mathcal{A}\to \mathcal{C}$ between cpo's is an isomorphism. By Corollary \ref{cor:Approx_fullsubcat_cpo}, $\mathcal{B}^\mathcal{A}\in \greeks$, as desired.\qed
				\end{itemize}
			\end{proofwithoutqed}
			
			
			\subsection{An Approximation System}\label{sec:sub:approxsystem_greeks}
					Thanks to Theorem \ref{thm:LU_approx_category}, we can take $\greeks$ as the approximation category for any application in which we wish to apply the version of AFT of \citet{CRS18ApproximationFixpointTheoryWell-FoundedSemanticsHigher-Order}. Depending on the specific language and semantics at hand, demanded by the application, we would define a different approximation system with $\greeks$. In this subsection, we present the approximation system for the language $\HOL$ and the semantics used in \citep{CRS18ApproximationFixpointTheoryWell-FoundedSemanticsHigher-Order} to tackle higher-order logic programs. 
					
					The language $\HOL$ is based on a type hierarchy $\mathbb{H}$ with base types $o$, the boolean type, and $\iota$, the type of individuals. The composite types are  morphism types obtained from $o$ and $\iota$. In particular, the types are divided into \emph{functional types} $\sigma:= \iota\mid\iota\to \sigma$, \emph{predicate types} $\pi:= o\mid \rho\to\pi$, and \emph{parameter types} $\rho:=\iota \mid \pi$. The semantics of the base types are defined as usual: $E_o:=\{\ltrue,\lfalse\}$ with the truth order $\lfalse\leq_t \ltrue$, and $E_\iota=D$ with the trivial order ($d_1\leq d_2$ iff $d_1=d_2$), where $D$ is some fixed domain for individuals. The semantics for composite types are defined following the Cartesian closed structure of $\poset$. For instance, the semantics of type $\bool \to \bool$ is simply the poset of functions from  $E_\bool$ to itself, \ie $E_{\bool\to\bool}:=(E_\bool \to E_\bool)$.
			
			Since the ultimate goal of this application is studying the well-founded semantics of higher order logic programs via AFT, we are interested in the approximation space of Herbrand interpretations. Since Herbrand interpretations fix the value assigned to symbols of functional types, we only need the approximation spaces for the semantics $E_\pi$, for all predicate types $\pi$. In other terms, we can focus on the smallest subset $S$ of $\Obj{\poset}$ containing $E_\pi$ for  all $\pi$, and closed under generalized product.
			
			Now the definition of a suitable approximation system for $S$ is very straightforward: we just have to define the approximation space $\mathit{App}(E_o)$, the set of exact elements $\mathcal{E}_o$, and the projection $\mathfrak{p}_o$. All the other elements are defined inductively following the Cartesian closed structure of  $\greeks$. We define: $\mathit{App}(E_o):=E_o\otimes E_o=\langle \{(\ltrue,\ltrue), (\lfalse, \ltrue), (\lfalse,\lfalse)\}, \leqp\rangle$;
			$\mathcal{E}_o=\{(\ltrue,\ltrue), (\lfalse, \lfalse)\}$;
			 and $\mathfrak{p}_o(\ltrue, \ltrue):=\ltrue$ and $\mathfrak{p}_o(\lfalse, \lfalse):=\lfalse$.
			Finally, given a vocabulary $V$ for $\HOL$ containing symbols of predicate type, and a program $\mathsf{P}$ over $V$, the approximation space of Herbrand interpretations of $\mathsf{P}$ is $\mathcal{H}_\mathsf{P}:=\mathit{App}\left( \Pi_{s\in V} E_{t(s)} \right)= \Pi_{s\in V} \mathit{App}(E_{t(s)})\in\Obj{\Approx}$, where $t(s)$ is the type of the symbol $s$. 
			
			This greatly simplifies the construction of \cite{CRS18ApproximationFixpointTheoryWell-FoundedSemanticsHigher-Order}. In particular,
			notice that the pairs of \emph{monotonte-antimonotone} and \emph{antimonotone-monotone} functions they defined are \emph{precisely} the elements of the exponential objects of $\greeks$. Moreover, by changing the base types and their semantics, this approximation system can be readily adapted to suit other applications.
			
			In conclusion, it is important to stress that we now have a clear concept of exactness: for the base type $\bool$ the exact elements are $\mathcal{E}_o=\{(\ltrue,\ltrue), (\lfalse, \lfalse)\}$, and for higher-order types, we follow Definition \ref{def:exact}. The work of 
			\citet{CRS18ApproximationFixpointTheoryWell-FoundedSemanticsHigher-Order} lacked a notion of exactness, making it impossible to determine whether a model is actually two-valued; they discussed this question in their future work section.  
			Let us illustrate on their example accompanying the discussion. 
			%

			\begin{example}\label{ex:exact_model_iclp18}
				Let $P$ be a program with the single rule $p(R)\lrule R$, where $p$ is a predicate of type $o\to o$ and $R$ is a variable of type $o$. The space of interpretations for $p$ is simply $\mathit{App}(E_{o\to o})=\mathit{App}(E_o)^{\mathit{App}(E_o)}$, \ie all the monotone functions from $\mathit{App}(E_o)=\{(\lfalse,\ltrue), (\lfalse,\lfalse), (\ltrue,\ltrue)\}$ to itself, as defined above. 
				By the semantics of \cite{CRS18ApproximationFixpointTheoryWell-FoundedSemanticsHigher-Order}, the meaning of this program is given by the interpretation $(I,J)$, where $I(p)(\ltrue,\ltrue)=J(p)(\ltrue,\ltrue)=J(p)(\lfalse,\ltrue)=\ltrue$, and $I(p)(\lfalse,\lfalse)=J(p)(\lfalse,\lfalse)=I(p)(\lfalse,\ltrue)=\lfalse$. Since $I\neq J$, $(I,J)$ is not exact according to the classical definition of AFT  \citep{DMT00ApproximationsStableOperatorsWell-FoundedFixpointsApplications}, even though we would expect to find a 2-valued model, \ie the one assigning to $p$ the identity function over $\{\lfalse,\ltrue\}$. Nevertheless, according to our definition, $(I,J)$ is indeed exact: it sends exacts of $E_o$ to exacts of $E_o$. Furthermore, by the approximation system we defined in this section, it is easy to see that $(I,J)$ represents $\mathfrak{p}_{o\to o}(I,J)=\mathcal{I}\in E_{o\to o}=(E_o\to E_o)$, where $\mathcal{I}(\ltrue)=\ltrue$ and $\mathcal{I}(\lfalse)=\lfalse$, as desired.
			\end{example}

			\subsection{A New Approximator}\label{sec:sub:approximator_greeks}
			
				As presented at the end of Section \ref{sec:AFT}, the approximator of \cite{CRS18ApproximationFixpointTheoryWell-FoundedSemanticsHigher-Order} does not provide the expected well-founded semantics for logic programs when there is an existential quantifier in the body of a rule.  
			In this subsection, we propose a new approximator that solves such issue. We achieve this by restricting the set of elements over which certain variables can range. In particular, variables that are arguments of a predicate being defined, \ie in the head of a rule, can range over all the elements of the approximation spaces of the corresponding types: we want to define the approximation of a higher order predicate also when the argument is an approximation. On the contrary, variables that appear exclusively in the body of a rule do not need to be approximated. In other words, a variable of type $\famtypes$ that is not argument of any predicate being defined, will be forced to range only over the set $ \exact{\famtypes}$ of exact elements.
			
			Before stating the new definition for the approximator, we briefly recall the full syntax of $\HOL$. We slightly modify the one presented in  \citep{CRS18ApproximationFixpointTheoryWell-FoundedSemanticsHigher-Order} to make it less heavy.
			
			The \emph{alphabet} of $\HOL$ consists of the following: predicate variables/constants of every predicate type $\pi$; individual variables/constants of type $\iota$; the equality constant $\approx$ of type $\iota\to\iota\to\bool$ for comparing individuals of type $\iota$; the conjunction constant $\wedge$ of type $\bool\to\bool\to\bool$; the rule operator constant $\leftarrow$ of type $\bool\to\bool\to\bool$; and the negation constant $\sim$ of type $\bool\to\bool$.
			
				Every predicate variable/constant and every individual variable/constant is a
				\emph{term}; if $\mathsf{E}_1$ is a term of type $\rho \to \pi$ and $\mathsf{E}_2$ a
				term of type $\rho$ then $(\mathsf{E}_1\ \mathsf{E}_2)$ is a term of type $\pi$.
				Every term is also an \emph{expression};
				if $\mathsf{E}$ is a term of type $\bool$ then $(\pnot \mathsf{E})$ is an expression of type $\bool$;
				if $\mathsf{E}_1$ and $\mathsf{E}_2$ are terms of type $\basedom$, then $(\mathsf{E}_1\approx \mathsf{E}_2)$ is an expression of type $\bool$.


				A \emph{rule} of $\HOL$ is a formula
				$\mathsf{p}\ \mathsf{R}_1 \cdots \mathsf{R}_n \lrule \mathsf{E}_1 \land \ldots \land \mathsf{E}_m$,
				where $\mathsf{p}$ is a predicate constant of type $\rho_1 \to \cdots \to \rho_n \to \bool$,
				$\mathsf{R}_1,\ldots,\mathsf{R}_n$ are distinct variables of types $\rho_1,\ldots,\rho_n$ respectively and
				the $\mathsf{E}_i$ are expressions of type $\bool$.
				The term $\mathsf{p}\ \mathsf{R}_1 \cdots \mathsf{R}_n$ is the \emph{head} of the rule and
				$ \mathsf{E}_1 \land \ldots \land \mathsf{E}_m$ is the \emph{body} of the rule. For the sake of simplicity, we often write $ \mathsf{E}_1, \ldots, \mathsf{E}_m$, in place of $ \mathsf{E}_1 \land \ldots \land \mathsf{E}_m$.
				A \emph{program} $\mathsf{P}$ of $\HOL$ is a finite set of rules. A \emph{state} $s$ of a program $\mathsf{P}$ is a function that assigns to each variable $\mathsf{R}$ of type $\rho$, an element of $\mathit{App}(E_\rho)$, if $\rho\neq \iota$, or an element of $E_\iota=D$, if $\rho=\iota$. We denote by $s[\mathsf{R}_1/d_1,\ldots ,  \mathsf{R}_n/d_n]$ a state that assigns to each $\mathit{R}_i$ the corresponding value $d_i$, and coincides with $s$ on the other variables.

			
		Finally, we provide the definition for the three-valued semantics of \cite{CRS18ApproximationFixpointTheoryWell-FoundedSemanticsHigher-Order} adapted to the above, slightly-modified definitions.

			\begin{definition}
				Let $\mathsf{P}$ be a program, ${\cal I}$ an interpretation of $\mathsf{P}$, and
				$s$ a state. The \emph{three-valued semantics} of expressions and bodies
				is defined as follows:
				\begin{enumerate}
					\item $\oldmwrst{\mathsf{c}}{\cal I}{s} = \mathsf{c}$, for every individual constant $\mathsf{c}$,
					\item $\oldmwrst{\mathsf{p}}{{\cal I}}{s} = {\cal I}(\mathsf{p})$, for every predicate constant $\mathsf{p}$,
					\item $\oldmwrst{\mathsf{R}}{{\cal I}}{s} = s(\mathsf{R})$, for every  variable $\mathsf{R}$,
					\item $\oldmwrst{(\mathsf{E}_1\ \mathsf{E}_2)}{{\cal I}}{s} = $ $\oldmwrst{\mathsf{E}_1}{\cal I}{s}$ $(\oldmwrst{\mathsf{E}_2}{\cal I}{s})$,
					\item $\oldmwrst{(\mathsf{E}_1 \bigwedge \mathsf{E}_2)}{\cal I}{s}=$ $ \glb_{\leq_t}\{\oldmwrst{\mathsf{E}_1}{\cal I}{s},$ $ \oldmwrst{\mathsf{E}_2}{\cal I}{s}\}$,
					\item $\oldmwrst{(\sim \mathsf{E})}{{\cal I}}{s} =  (\oldmwrst{\mathsf{E}}{{\cal I}}{s})^{-1}$, with $(\mathbf{t},\mathbf{t})^{-1}\!=\!(\mathbf{f},\mathbf{f})$, $(\mathbf{f},\mathbf{f})^{-1}\!=\!(\mathbf{t},\mathbf{t})$
					and $(\mathbf{f},\mathbf{t})^{-1}\!=\!(\mathbf{f},\mathbf{t})$,
					\item $\oldmwrst{(\mathsf{E}_1\approx \mathsf{E}_2)}{{\cal I}}{s} = \begin{cases}
						(\mathbf{t},\mathbf{t}),  & \text{if } \oldmwrst{\mathsf{E}_1}{{\cal I}}{s} = \oldmwrst{\mathsf{E}_2}{{\cal I}}{s} \\
						(\mathbf{f},\mathbf{f}), & \text{otherwise}
					\end{cases}$.
				\end{enumerate}
				where $\leq_t$ is the \emph{truth order} defined by $\mathbf{f}\leq_t \mathbf{u}\leq_t \mathbf{t}$.
			\end{definition}
		
%
		As already explained, we want to restrict a variable appearing only in the body of a rule to range over the elements of $\exact{\famtypes}$, with $\famtypes$ being the type of the variable. We call a state $s$ \emph{exact} if for all variables $\mathsf{R}$ of type $\famtypes\neq\iota$, it holds that  $s(\mathsf{R})\in\exact{\famtypes}$.
We denote by			$\mathcal{S}$ the set of exact states.

		We have now all the elements to introduce the new approximator, i.e.\ the \emph{three-valued immediate consequence operator}. For the sake of simplicity, in the following we define $\mathit{App}(E_\iota):=E_\iota=D$, even though $E_\iota\notin S$ and has no associated approximation space.
		
			
			\begin{definition}\label{def:approximator}
				Let $\mathsf{P}$ be a program. The \emph{three-valued immediate consequence operator}
				${\Psi}_\mathsf{P} : {\cal H}_\mathsf{P} \to {\cal H}_\mathsf{P}$ is defined 
				for every predicate constant $\mathsf{p} : \rho_1 \to \cdots \to \rho_n \to \bool$ in $\mathsf{P}$, and for
				all $d_1 \in \mathit{App}(E_{\rho_1}),\ldots, d_n \in \mathit{App}(E_{\rho_n})$, as:
				${\Psi}_{\mathsf{P}}({\cal I})(\mathsf{p})\ d_1 \cdots d_n =
				\lub_{\leq_t}\{
				\oldmwrst{\mathsf{E}}{{\cal I}}{s[\mathsf{R}_1/d_1,\ldots ,  \mathsf{R}_n/d_n]} \mid
				\mbox{$s\in \mathcal{S}$ and
					$(\mathsf{p}\ \mathsf{R}_1 \cdots \mathsf{R}_n \lrule \mathsf{E})$ in $\mathsf{P}$}\}$.
			\end{definition}
		
		With Definition \ref{def:approximator}, we solve the issue linked to existential quantifiers. Let us  review the example presented in Subsection \ref{sec:AFT}. We considered a program $\mathsf{P}$ with just one rule $\mathsf{p} \lrule \mathsf{R}\wedge \sim \mathsf{R}$, where $\mathsf{p}$ is a predicate constant of type $\bool$, and $\mathit{R}$ is a variable of type $\bool$. Observe that in this case the space of Herbrand interpretations is just $\mathcal{H}_\mathsf{P}:= \mathit{App}(E_{\bool})\in\Obj{\Approx}$, as we are only interested in the interpretation of the predicate $\mathsf{p}$. For all interpretations $\mathcal{I}\in \mathcal{H}_\mathit{P}$, we have
		
		\begin{equation}\label{eq:approximator}
			\begin{split}
				{\Psi}_{\mathsf{P}}({\cal I})(\mathsf{p}) &=
				\lub_{\leq_t}\{
				\oldmwrst{\mathsf{E}}{{\cal I}}{s} \mid
				s\in \mathcal{S} \text{ and }
					(\mathsf{p} \lrule \mathsf{E}) \text{ in } \mathsf{P}\}=
					\\ &= \lub_{\leq_t}\{
					\oldmwrst{\mathsf{R}\wedge \sim \mathsf{R}}{{\cal I}}{s} \mid
					s\in \mathcal{S} \}= \lub_{\leq_t}\{
					s(\mathsf{R})\wedge s(\mathsf{R})^{-1} \mid
					s\in \mathcal{S} \}=
						\\ &= \lub_{\leq_t}\{ (\mathbf{f},\mathbf{f})\wedge (\mathbf{t},\mathbf{t}), (\mathbf{t},\mathbf{t})\wedge (\mathbf{f},\mathbf{f})\}= (\mathbf{f},\mathbf{f}).
			\end{split}
		\end{equation}
		
		Notice that for this specific program $\mathsf{P}$, both the approximator $	{\Psi}_{\mathsf{P}}$ and the old version from \citep{CRS18ApproximationFixpointTheoryWell-FoundedSemanticsHigher-Order} do not depend on the interpretation $\mathcal{I}$, but only on the   states. While the approximator of \cite{CRS18ApproximationFixpointTheoryWell-FoundedSemanticsHigher-Order} considered any possible state, even the one sending $\mathsf{R}$ to $ (\mathbf{f},\mathbf{t})$, ${\Psi}_{\mathsf{P}}$ takes into account only exact states, i.e.\ $\mathsf{R}$ can only be sent to an element of $\exact{\bool}=\{(\mathbf{f},\mathbf{f}),(\mathbf{t},\mathbf{t})\}$. This limitation removes the formula $(\mathbf{f},\mathbf{t})\wedge (\mathbf{f},\mathbf{t})=(\mathbf{f},\mathbf{t})$ from the least upper bound computation in \eqref{eq:approximator}, which was the one causing the evaluation of $\mathsf{p}$ being $(\mathbf{f},\mathbf{t})$ for the approximator of \cite{CRS18ApproximationFixpointTheoryWell-FoundedSemanticsHigher-Order}. 
		
		Since the approximator $	{\Psi}_{\mathsf{P}}$  does not depend on the interpretation, it is immediate to see that the \emph{well-founded} operator $\mathcal{S}_{{\Psi}_{\mathsf{P}}}$ coincides with the approximator ${\Psi}_{\mathsf{P}}$ for all $(I,J)\in \mathcal{H}_\mathsf{P}$:
		
		\begin{equation*}
			\mathcal{S}_{{\Psi}_{\mathsf{P}}}(I,J)=
			\big(\lfp({\Psi}_{\mathsf{P}}(\cdot, J)_1), \lfp({\Psi}_{\mathsf{P}}(I,\cdot)_2)\big)
			=\big({\Psi}_{\mathsf{P}}(\cdot, J)_1,
			 {\Psi}_{\mathsf{P}}(I,\cdot)_2\big)
			={\Psi}_{\mathsf{P}}(I,J).
		\end{equation*}
		
		It follows that $\mathcal{S}_{{\Psi}_{\mathsf{P}}}$  does not depend on the interpretation either. Thus, the least fixpoint of $\mathcal{S}_{{\Psi}_{\mathsf{P}}}$, which corresponds to the well-founded model of $\mathsf{P}$, is just the interpretation sending $\mathsf{p}$ to $(\mathbf{f},\mathbf{f})$, resulting in a sensible account for the well-founded sementics. Moreover, observe that this interpretation is also exact by our new definition, and it corresponds to the unique exact stable model of the program $\mathsf{P}$.


%

	In the remainder of this section, we present two examples that highlight the importance of enabling the application of approximate objects to approximate objects.

		\begin{example}\label{ex:graph}
			Consider an undirected graph given by a predicate $\verb|node|\colon \iota\to\bool$, containing all the nodes of the graph, and a predicate $\verb|edge|\colon \iota\to\iota\to\bool$ defining the edge relation, which we assume to be symmetric. Some nodes of the graph are \emph{marked}. We call a set of nodes $S$ a \emph{covering} if for every marked node $n$ there exists a node in $S$ with an edge to $n$. Now, suppose that a Player can modify the set of marked nodes by swapping a marked node with a neighbouring, unmarked node. The goal of the Player is reached when the set of marked nodes is a covering. At that point, the game is over and the Player cannot swap nodes anymore.

			The key predicates of our example are contained in Listing \ref{listing:graphgame}. We have defined them  in terms of a time parameter $\verb|T|$ of type $\iota$, assuming that the Player can only do one swap at a time. In particular,  here are the signatures and meanings of the main predicates of Listing \ref{listing:graphgame}: $\verb|swap|\colon \iota \to\iota\to\iota\to\bool$ indicates whether at a certain time, two nodes are swapped; $\verb|marked|\colon \iota\to\iota\to\bool$ represents the set of marked nodes at a specific time (Lines \ref{line:marked1} to \ref{line:marked2}); $\verb|covering|\colon \iota\to(\iota\to\bool)\to\bool$ tells whether at a certain time a set of nodes is a covering (Lines \ref{line:covering1}, and \ref{line:covering2}); and $\verb|gameOver|\colon \iota$ expresses whether the game is over at a certain time (Line \ref{line:gameOver}).
			
				\begin{lstlisting}[caption=Graph Game.,label=listing:graphgame]
				% We define predicates to add/remove nodes to/from marked set.			
				add T X  :- swap T X Y 
				remove T Y :- swap T X Y
				
				% We define the set of marked nodes based on the marked nodes at the previous time point and the last swap.
				marked T X :- succ T' T, ~(gameOver T), marked T' X, |\\\newlineInListing|~(remove T' X). |\label{line:marked1}|
				marked T X :- succ T' T, ~(gameOver T), add T' X.
				marked T X :- succ T' T, gameOver T', marked T' X. |\label{line:marked2}|
				
				% We define what a covering of a set of nodes is.
				nextTo S X :- S Y, edge Y X.
				nonsubset S Q :- S X, ~(Q X).
				subset S Q :-  ~(nonsubset S Q).
				ncovering T S :- marked T X, ~(nextTo S X). |\label{line:covering1}|
				covering T S :- subset S node, ~(ncovering T S). |\label{line:covering2}|
				
				% We define when the game ends.
				gameOver T :- covering T (marked T).  |\label{line:gameOver}|
			\end{lstlisting}

				In the formalization of \citet{CRS18ApproximationFixpointTheoryWell-FoundedSemanticsHigher-Order}, natural numbers are not taken into account. Thus, we regard the time variable $\verb|T|$ as an individual variable of type $\iota$ and we limit our example to only three time points, expressed by the individual constants $\verb|a|, \verb|b|$, and $\verb|c|$, related by the successor relation $\verb|succ|\colon \iota\to\iota\to\bool$, as expressed in Listing \ref{listing:timepoints}. In the same Listing, we also instantiate the nodes and edges of the graph that we chose for this example, and the initial set of marked nodes, i.e. $\verb|marked a|$.

				\begin{lstlisting}[caption={Instantiation of time points,  graph's nodes and edges, and initial set of marked nodes.},label=listing:timepoints]
				% Time points a, b, and c: b successor of a, and c of b.
				time a.  |\label{line:timea}|
				time b.  
				time c.  
				succ b a. 
				succ c b.  |\label{line:successorcb}|
				% Nodes.
				node x. 
				node y. 
				node z. 
				node u. 
				node v. 
				% Edges.
				edge x y.
				edge x z.
				edge x u.
				edge z v.
				edge u v.
				% Marked nodes at the start of the game (time point a).
				marked a y.
				marked a u.
					\end{lstlisting}
				
		We are only missing the swaps the Player makes at each of the three time points. For example, we could have those listed in Listing \ref{listing:swaps}. 
		
			\begin{lstlisting}[caption=Swaps.,label=listing:swaps]
			% Nodes x, y, z, u, and v, and swaps at the time points a, b, and c.
			swap a v u.
			swap b x y.  
			swap c z x. 
		\end{lstlisting}
	
	By joining Listings \ref{listing:graphgame}, \ref{listing:timepoints}, and \ref{listing:swaps}, we obtain a program $\mathsf{P}$ encoding a specific run of the game.
		
		Using the machinery of AFT, we can easily find the well-founded, the stable, the Kripke-Kleene, and the supported models  of  $\mathsf{P}$.
		To obtain the well-founded operator, we compute the least fixpoint of the well-founded operator of the approximator contained in Definition \ref{def:approximator}, \ie the least fixpoint of $\subscript{\mathcal{S}}{{\Psi}_{\mathsf{P}}}\colon (x,y)\mapsto (\subscript{S}{{\Psi}_{\mathsf{P}}}(y), \subscript{S}{{\Psi}_{\mathsf{P}}}(x))$, where $\subscript{S}{{\Psi}_{\mathsf{P}}}\colon x\mapsto \lfp({{\Psi}_{\mathsf{P}}}_1(\cdot , x))$ is the stable operator\footnote{The well-founded, and the stable operator have been briefly introduced in Section \ref{sec:AFT} of the Preliminaries.}. Since the well-founded operator is monotone, to find its least fixpoint it is sufficient to repeatedely apply the operator starting from the bottom element of its domain, namely the interpretation sending every predicate constant to the bottom element of the respective approximation space. 
		Notice that during the first iterative applications of the well-founded operator, the predicates $\verb|marked|$, $\verb|covering|$, and $\verb|gameOver|$ are being defined only for the first time points, \ie they are partially defined. In other words, the three-valued interpretations that we obtain from the first computations leading to the well-founded fixpoint, send the aforementioned predicates to approximate objects of the respective approximation spaces. Since $\verb|marked|$, $\verb|covering|$, and $\verb|gameOver|$ are all defined by mutual induction,  we are forced to apply an approximate object on another approximate object. In particular, in Line \ref{line:gameOver} of Listing \ref{listing:graphgame}, for the definition of $\verb|gameOver|$, the predicate $\verb|covering|$ is applied on $\verb|marked|$. Only when the fixpoint is reached, all the predicates being defined will be exact, i.e.\ two-valued.

		In Listing \ref{listing:swaps}, we have provided a specific set of swaps the Player makes. We can obtain a more general setup by using choice rules to define the predicate $\verb|swap|$, as we do in Listing \ref{listing:choicerules}.

			\begin{lstlisting}[caption=Choice Rules.,label=listing:choicerules]
			% Choice rules: the user can swap one edge with a neighbour at each time. 
			swap T X Y :- node X, node Y, time T, ~(nswap T X Y).  |\label{line:swap1}|
			nswap T X Y :- node X, node Y, time T, ~(swap T X Y). |\label{line:swap3}|
			% First element must be in the marked set, second cannot 
			:- swap T X Y, ~(marked T X). |\label{line:swap4}|
			:- swap T X Y, marked T Y.
			% At most one swap at a time.
			:- swap T X Y, swap T X' Y', ~(X'=X).
			:- swap T X Y, swap T X' Y', ~(Y'=Y).
			% Can only swap neighbors.
			:- swap T X Y, ~(edge X Y). |\label{line:swap2}|
		\end{lstlisting}
		
		By joining Listings \ref{listing:graphgame}, \ref{listing:timepoints}, and \ref{listing:choicerules}, we obtain another program $\mathsf{P}'$, and we can again compute the models of interest via AFT. 
		In particular, now each stable model corresponds to a possible run of the game with starting set of  marked nodes $\verb|marked a|$.
		Notice that, because of the choice rules in Listing \ref{listing:choicerules}, the well-founded model of $\mathsf{P}'$ leaves most of the predicates undefined.
		
		\end{example}

		\begin{example}\label{ex:manifacturer}
			Let us consider a manufacturing company that aims at growing and diversifying its production. We represent raw materials with individual constants of type $\iota$, and finished products with predicate constants of type $\iota\to \bool$, such that if $\verb|P|$ is any finished product, and $\verb|M|$ is any raw material, then $\verb|P M|$ is true if and only if $\verb|M|$ is necessary to craft $\verb|P|$. We denote by $\verb|materials|\colon \iota\to \bool$ the predicate corresponding to the set of all raw materials, and by $\verb|products|\colon (\iota\to\bool)\to\bool$ the predicate corresponding to the set of all finished products.
			
			
			We want to define a predicate $\verb|production|\colon \iota\to \iota \to \bool$ (Lines \ref{line:production1} and \ref{line:production2} of Listing \ref{listing:company}) that indicates which raw materials the company has to acquire for production at a certain time: $\verb|production T M|$ is true if and only if the company acquires the material $\verb|M|$ at time $\verb|T|$. As in Example \ref{ex:graph}, we regard the time variable $\verb|T|$ as an individual variable of type $\iota$ and we limit our program in Listing \ref{listing:company} to only three time points, related by the successor relation (Lines \ref{line:time1} to \ref{line:successor2}).
			 We fix the initial set $\verb|production a|\colon \iota\to \bool$ of materials the company starts with. At each time point, the company decides which new materials to acquire: we encode the information about these potential new ingredients with the predicate $\verb|candidates|\colon \iota\to(\iota\to\bool)\to\bool$ (Lines \ref{line:candidates1} and \ref{line:candidates2}), which takes as argument a time point, \ie an individual variable, and a set of materials, \ie a predicate of type $\iota\to\bool$. The selection of new materials the company takes into consideration obeys a few criteria:
			 \begin{enumerate}
			 	\item \textsl{Maximize profit}: products necessitating more raw materials to be crafted require more expertise and more capital to invest, but they provide more profit. Hence, as time progresses, the company aims at products more and more complex: at time $\verb|T|$, a set $\verb|P|$ of raw materials is a candidate if it can be covered by sets corresponding to some finished products of complexity $\verb|T|$ (Lines \ref{line:subclique}, \ref{line:notcliquecovered}, and \ref{line:candidates1}). We assume a constant predicate $\verb|complexity|\colon (\iota\to \bool)\to \iota\to \bool$ indicating the complexity of a product is given. 
			 	\item \textsl{Cautiousness}: producing items using only new materials may be risky and time consuming,  as the manifacturing team has to acquire novel knowledge, and new suppliers for the raw materials need to be found.  Hence, the newly accepted products are required to share at least one raw material with a product in production at the prevous time point (Lines \ref{line:intersection}, \ref{line:subclique}, \ref{line:notcliquecovered}, and \ref{line:candidates1}).
			 	\item \textsl{Efficient growth}: as time passes and the company produces more complex items, older, simpler products can be put out of production. This is done gradually: if in $\verb|production T|$ there is still some material that is not needed to craft any product of size $\verb|T|$ or $\verb|succ T|$, then such material will not be in production at the following time point; otherwise, all materials that are not needed to craft any product of size $\verb|succ T|$ are dropped out of production at time $\verb|succ T|$ (Lines \ref{line:candidates2}). In other words, all the products with the lowest complexity are dropped.

			 \end{enumerate}
		 
		 Finally, $\verb|production|$ is just the union of all the candidates (Line \ref{line:production1}). If there are no candidates at a certain time point, the production does not vary at the next time point, and the company ends its expansion (Lines \ref{line:existcandidate}, and \ref{line:production2}).
			
				\begin{lstlisting}[caption=The growth of the manifacturing company.,label=listing:company]
				% a, b, and c are time points, b is the successor of a, and c of b.
				time a.  |\label{line:time1}|
				time b.  |\label{line:time2}|
				time c.  |\label{line:time3}|
				succ b a.  |\label{line:successor1}|
				succ c b.  |\label{line:successor2}|
				% intersect Q R is true if the intersection between Q and R is non-empty.
				intersects Q R :- Q X, R X.  |\label{line:intersection}|
				% subset Q R is true (nonsubset Q R is true) if R is (is not) a subset of Q.
				nonsubset Q R :- Q X, ~(R X).|\label{line:nonsubset}|
				subset Q R :-  ~(nonsubset Q R).
				
				% There exists a product P made of T raw materials, all belonging to S, some belonging to production T', and including M.
				existsubprod T S M :- products P, complexity P T, subset P S, |\\\newlineInListing|P M, intersects P (production T'), succ T T'.  |\label{line:subclique}|
				% S cannot be covered by products of size T if there exists a material M in S that is never part of a subproduct of S.
				notcovered T S :- S M,  ~(existsubprod T S M), time T.  |\label{line:notcliquecovered}|
				% S is a candidate at time T if it can be covered and it is a set of materials (maximize profit, and cautiousness).
				candidates T S :- ~notcovered T S, time T, subset S materials.  |\label{line:candidates1}|
				% S is a candidate if it is a product of size T', in production at time T', and if in production T' there were still raw materials only needed for products of size different than T' (efficient growth).
				candidates T S :- ~(candidates T' (production T')), products P, |\\\newlineInListing|subset P (production T'), complexity P T', succ T T'.  |\label{line:candidates2}|
				% There exists a candidate at time T.
				existcandidate T :- candidates T S.  |\label{line:existcandidate}|
				% A material is needed for production at time T if it is in one of the candidate sets at time T.
				production T M :- candidates T S, S M.  |\label{line:production1}|
				% If there are no candidates at time T, then the production remains the same.
				production T M :- ~(existcandidate T), production T' M,  |\\\newlineInListing|succ T T'.|\label{line:production2}|
			\end{lstlisting}
	
	Note that the symbols $\verb|Q|$, $\verb|R|$, and $\verb|S|$, highlighted in red, are predicate variables of type $\iota \to \bool$ that are used in a higher-order style in the rules in Lines \ref{line:intersection}, \ref{line:nonsubset}, \ref{line:notcliquecovered}, and \ref{line:production1}.
	
	Similarly to what happens for Example \ref{ex:graph}, since the predicates $\verb|candidates|$ and $\verb|production|$ are defined by mutual induction, they will be fully defined only at the end of the least fixpoint construction.
	However, in order to define $\verb|candidates|$ we need to apply it to $\verb|production|$ (Line \ref{line:candidates2}). In particular, before reaching the fixpoint, we will need to apply an approximate object, namely $\verb|candidates|$, on another approximate object, \ie $\verb|production|$. 
	
\end{example}

			\section{Conclusion}\label{sec:conclusions}
			
			We introduced a novel theoretical framework that provides a mathematical foundation for using the machinery of AFT on higher-order objects. 
			In particular, we defined \emph{approximation categories} and \emph{approximation systems}: they employ the notion of  Cartesian closedness to inductively construct a hierarchy of approximation spaces for each semantics of the types of a given (higher-order) language. This approach solves the issue of applying approximate objects onto approximate objects and ensures that the approximation spaces have the same mathematical structure at any order of the hierarchy, enabling the application of the same AFT techniques at all levels. Moreover, we defined exact elements of a higher-order approximation space, together with a projection function.
			This is a non-trivial definition and it is fundamental to obtain a sensible AFT framework, \ie a framework in which we can determine when an object, and in particular a model, is two-valued, and retrieve the elements that are being approximated.
			
			Despite seeming complicated at first, the use of CT not only provides a solid, formal mathematical foundation to work with, but also allows to reduce the complexity of proofs. 
			The inductive nature and generality of the definition of an approximation system make it extremely easy to adapt the framework to different languages, types, and semantics, as we only have to modify the base elements of the induction. Such generality enables extending different existing versions of AFT while capturing their common underlying characteristics, as we have shown for standard AFT and
			the extension of consistent AFT of \citet{CRS18ApproximationFixpointTheoryWell-FoundedSemanticsHigher-Order}. Moreover, concerning the latter version of AFT mentioned, we have resolved its problematic features. In particular, we provided a novel approximator which behaves properly, even on programs with existential quantifiers in the body of rules, and we defined the concept of exactness, previously missing in the work of \citet{CRS18ApproximationFixpointTheoryWell-FoundedSemanticsHigher-Order}, allowing to consider exact stable models.

			As far as future work and developments are concerned, it is paramount to notice that by systematically extending present (and possibly future) versions of AFT to the higher-order setting, this paper further enriches the vast body of algebraic results on AFT. In particular, this allows us to make all the already-existing formal results regarding AFT readily available in a higher-order AFT context. This includes, but it is not limited to,
			stratification results \citep{VGD06SplittingoperatorAlgebraicmodularityresultslogics,BC21StratificationApproximationFixpointTheoryApplicationActive}, grounded fixpoints \citep{BVD15Groundedfixpointstheirapplicationsknowledgerepresentation}, well-founded induction \citep{DV07Well-FoundedSemanticsAlgebraicTheoryNon-monotoneInductive}, and strong equivalence \citep{T06Stronguniformequivalencenonmonotonictheories-}. Moreover, in light of this newly established bridge to the higher-order environment, one could explore the possibilities within the application fields where AFT previously succeeded,
			such as abstract argumentation \citep{S13Approximatingoperatorssemanticsabstractdialecticalframeworks,B19WeightedAbstractDialecticalFrameworksthroughLens},  active integrity constraints \citep{BC18Fixpointsemanticsactiveintegrityconstraints}, stream reasoning \citep{A20Fixedpointsemanticsstreamreasoning},  integrity constraints for the semantic web \citep{BJ21FixpointSemanticsRecursiveSHACL}, and Datalog \citep{pollaci2025}.
			
			Lastly, it may be of interest to research how the developed higher-order semantics and language $\HOL$ presented in Section \ref{sec:iclp18} relate to Hilog and Prolog with meta-predicates \citep{CKW93HILOGFoundationHigher-OrderLogicProgramming}. In particular, $\HOL$ and Prolog show two rather different natures: while Prolog is procedural and intensional, the language $\HOL$  provides a declarative and extensional approach. This is indeed not trivial to obtain for the higher-order setting, as it was also pointed out by \cite{RS18ExtensionalSemanticsHigher-OrderLogicProgramsNegation}. 

			\section*{Competing Interests}
		The authors declare none.

\bibliographystyle{tlplike}
		
			\bibliography{bb_tex/bb_refs-to-appear,bb_tex/bb_refs,otherrefs}

				\end{document}